\documentclass[12pt,a4paper]{amsart}

\usepackage[DIV=12,BCOR=12mm]{typearea}

\usepackage{xcolor}
\usepackage{microtype}
\usepackage{csquotes}
\usepackage[english]{babel}
\usepackage{todonotes}
\usepackage{nicematrix}
\usepackage{booktabs}
\usepackage{multirow}
\usepackage{tikz}
\usetikzlibrary{calc}
\usepackage{tikz-cd}
\usepackage{braket}
\usepackage{physics}
\usepackage{cancel}

\usepackage[widespace]{fourier}
\usepackage[colorlinks,linkcolor=blue,citecolor=teal]{hyperref}
\usepackage{thmtools}
\usepackage[capitalise]{cleveref}

\usepackage{cite}

\newcommand{\galg}{\mathfrak{g}}
\newcommand{\lieg}{\galg}
\newcommand{\halg}{\mathfrak{h}}
\newcommand{\Sl}{\mathfrak{sl}}
\newcommand{\Sp}{\mathfrak{sp}}
\newcommand{\gl}{\mathfrak{gl}}
\newcommand{\K}{\mathbb{K}}
\newcommand{\R}{\mathbb{R}}
\newcommand{\CC}{\mathbb{C}}
\newcommand{\PP}{\mathbb{P}}
\newcommand{\N}{\mathbb{N}}

\usepackage{color}

\usepackage[normalem]{ulem}

\definecolor{RED}{named}{red}

\DeclareMathOperator{\SO}{SO}
\DeclareMathOperator{\SL}{SL}
\DeclareMathOperator{\DD}{D}
\DeclareMathOperator{\Span}{span}
\DeclareMathOperator{\Pol}{Pol}

\DeclareMathOperator{\Sym}{Sym}
\DeclareMathOperator{\adj}{ad}

\DeclareMathOperator{\id}{Id}
\DeclareMathOperator{\Cas}{Cas}

\theoremstyle{plain}
\newtheorem{theorem}{Theorem}[section]
\newtheorem{proposition}[theorem]{Proposition}
\newtheorem{corollary}[theorem]{Corollary}
\newtheorem{lemma}[theorem]{Lemma}
\theoremstyle{definition}
\newtheorem{definition}[theorem]{Definition}

\theoremstyle{remark}
\newtheorem{remark}[theorem]{Remark}
\newtheorem{example}[theorem]{Example}

\newtheorem*{conjecture*}{Conjecture}

\author[G. Gubbiotti]{Giorgio Gubbiotti}
\author[D. Latini]{Danilo Latini}
\author[B. van Geemen]{Bert van Geemen}
\address[G. Gubbiotti]{Dipartimento di Matematica ``Federigo Enriques'',
        Universit\`a degli Studi di Milano, Via C. Saldini 50, 20133
        Milano, Italy \& INFN Sezione di Milano, Via G. Celoria 16,
        20133 Milano, Italy}
\email{giorgio.gubbiotti@unimi.it}
\address[D. Latini]{Dipartimento di Matematica ``Federigo Enriques'',
        Universit\`a degli Studi di Milano, Via C. Saldini 50, 20133
        Milano, Italy \& INFN Sezione di Milano, Via G. Celoria 16,
        20133 Milano, Italy}
\email{danilo.latini@unimi.it}
\address[B. van Geemen]{Dipartimento di Matematica ``Federigo Enriques'',
        Universit\`a degli Studi di Milano, Via C. Saldini 50, 20133
        Milano, Italy}
\email{lambertus.vangeemen@unimi.it}

\numberwithin{equation}{section}

\allowdisplaybreaks

\title{A novel chain of Lie algebras and its coalgebra symmetry}

\begin{document}

\begin{abstract}
    We study a novel $n(n+1)/2$-dimensional non-semisimple Lie algebra
    $\galg_n$, a generalisation of both $\Sl_2(\K)$ and the two-photon Lie
    algebra $\mathfrak{h}_6$. We investigate its properties, including its
    structure, representations, and its Casimir elements. In particular, we
    prove that there exists only one non-trivial Casimir polynomial of degree
    $n$ given by the determinant of an $n\times n$ symmetric matrix. We then
    associate this Lie algebra to a hierarchy of Hamiltonian systems with
    integrability properties depending on $n$, and describe their first
    integrals as sums of squares of linear combinations of the components of
    the angular momentum. In particular, we obtain that these systems are
    integrable for $n=2$, quasi-integrable for $n=3$, and of
    Poincar\'e--Lyapunov--Nekhoroshev type for $n\geq4$.
\end{abstract}
\maketitle
%


%
\setcounter{tocdepth}{1}
\tableofcontents

\section{Introduction}

\noindent Lie algebras play a pivotal role in the study of Integrable Systems,
serving as a fundamental tool for understanding and classifying these
mathematical structures. Crucial to the advancement of this field was the
development of Quantum Mechanics. Indeed, the Lie algebraic approach to
integrable systems was initiated in Hendrik Casimir's Ph.D.
thesis~\cite{Casimir1931} on rigid body motion in Quantum Mechanics, and
later continued with the seminal paper by Fock~\cite{Fock1935} on the
$\SO(4)$ symmetry of the hydrogen atom. A similar finding was made around
of the same time by Demkov~\cite{Demkov1956} and Fradkin~\cite{Fradkin1965}
about the harmonic oscillator, whose underlying symmetry algebra was found
to be related to the special unitary algebra, see also~\cite{DGL_fradkin}.
This research left an everlasting mark in the history of the topic
as it is witnessed by the famous book of Herman Weyl on the
subject~\cite{Weyl1950theory}. On the other hand, the importance of Lie
algebras in classical mechanics was established only later, mostly
through the fundamental work of Olshanetsky and Perelomov. In their most
famous paper~\cite{OlshanetskyPerelomov1976} Olshanetsky and Perelomov
established a connection between semi-simple Lie algebras and
Calogero--Moser like integrable Hamiltonian systems~\cite{Calogero1971},
see also the book~\cite{Perelomov1990book}. 

The systematisation of these ideas led to the so-called algebraic approach
to Integrable Systems. This field uses tools from Lie or Poisson algebra
theory to construct new integrable systems and explicitly solve existing
ones. The first problem, also known as the ``search problem'' is of
fundamental importance in the theory of Integrable Systems and in
Mathematical Physics in general, since Integrable Systems are ubiquitous in
applications, as they arise as a limiting form of a broad families of
nonlinear systems~\cite{Calogero1991}.  Given a ``generic system'', it is
possible to obtain information on its behaviour using perturbation theory
of an integrable system, see e.g.~\cite[Chap.\ 11]{Arnold1997}.

The strength of most algebraic methods in Integrable Systems is that, in
contrast to various existing direct methods, see e.g.\
\cite{Fris1965,PostWinternitz2011, Kalninsetal2002, Kalninsetal2003,
Gravel2004, PostWinternintz2015general, Drach1935, KarloviniRosquist2000,
Karlovini_etal2002, GrammaticosDorizziPadjen1982,
RamaniDorizziGrammaticos1982, DorizziGrammaticosRamani1983,
Quispel_etal2023, vanderKamp_etal2021homogeneous, PrelleSinger1981,
Schlomiuk1993proc, Schlomiuk1993, GubLatDrach}, 
they can be employed to construct Integrable Systems in
arbitrary dimensions. In contrast, direct methods are often limited by
cumbersome computations as the number of degrees of freedom increases.

As in any top-down approach, the main limitation of algebraic methods in
Integrable Systems is our knowledge of the properties of a given mathematical
structure. For this reason, the techniques based on Lie algebras in most of the
cases are built on (semi)simple Lie algebras, which are essentially completely
understood over the complex field since the times of W.\
Killing~\cite{Killing1890} and \'E.\ Cartan~\cite{Cartan1909}. When the real
field is considered, for instance for applications in Classical Mechanics, the
situation is more involved, since complex simple Lie algebras have multiple
possible real forms~\cite{Onishchik2004}.  Besides the results on simple Lie
algebras, the classification of (all) real Lie algebras is nowadays limited up
to dimension $6$, with an history going back to Luigi Bianchi~\cite{Bianchi}
and completed only about 70 years later by Turkowski~\cite{Turkowski1990} with
many contributions given by
Mubarakzyanov~\cite{Mubarakzyanov1963a,Mubarakzyanov1963b,Mubarakzyanov1963c}.
For more details see the very complete introduction on the topic in~\cite[\S
2]{Popovich2003}. Classification in higher dimension becomes impractical as the
number of Lie algebras grows exponentially (see the estimates
in~\cite{KirillovNeretin1987}), so usually one confines the study to Lie
algebras with interesting properties, where of course the interest depends on
the intended application.

This paper is devoted to the study of a novel $n(n+1)/2$-dimensional
non-semimple Lie algebra, which we call $\galg_{n}$ for $n\geq 2$. Up to
isomorphism, the Lie algebra $\galg_n$ is a generalisation of another
non-semisimple Lie algebra called the two-photon Lie algebra
$\mathfrak{h}_6$~\cite{Zhang_et_al1990}, which arises as $\galg_3$. Moreover,
it is well-known that the two-photon Lie algebra admits as Levi factor a copy
of the three-dimensional special linear algebra $\Sl_2(\K)$, which arises in
our construction as $\galg_2$. The Lie algebra $\galg_n$ generalises this
property, and in particular, we will show that for every fixed $n$ it contains
as Lie subalgebras all the Lie algebras $\galg_k$ with $2\leq k\leq n-1$. The
Lie algebra $\galg_n$ was first proposed in~\cite{GLT_coalgebra} in the context
of the application of Lie algebras to \emph{discrete time} integrable systems
through the so-called coalgebra symmetry
approach~\cite{Ballesteros_et_al_1996,BallesterosRagnisco1998,Ballesteros_et_al2009}.
Roughly speaking, the coalgebra symmetry approach is an algebraic machinery
that given a (classical) one degree of freedom Hamiltonian system associated to
a Lie--Poisson algebra turns it into one with $N$ degrees of freedom, and many
first integrals are obtained from the Casimir elements of the underlying
Lie-Poisson algebra. In some particular cases, these first integrals are enough
to give both integrability and even superintegrability of the $N$ degrees of
freedom Hamiltonian system.  In~\cite{GLT_coalgebra} it was proposed that the
Lie algebra $\galg_{n}$ could be useful to explain integrability of both
discrete and continuous time systems, making it a ``good'' candidate Lie
algebra to be studied in arbitrary dimension, i.e.\ arbitrary $n$. However, the
idea given in~\cite{GLT_coalgebra} requires the computation of the Casimir
elements of the Lie algebra $\galg_n$ for generic $n$. This is a difficult
problem to tackle because the algebra is non-semisimple and for this kind of
Lie algebras the standard ways of computing the Casimir elements rapidly become
untreatable because of the complexity of the calculation involved.

In this paper, we address the problem of determining the Casimir elements of
the Lie algebra $\galg_n$ by initially employing the vector field approach
based on its coadjoint representation. We then rigorously prove that the
obtained polynomial is indeed the Casimir invariant for any $n \geq 2$, within
a framework grounded in the representation theory of Lie
algebras~\cite{FultonHarris1991}. We are able to show that there exists, up to
a scalar multiple, only one non-trivial Casimir polynomial of degree $n$. This
result generalise the well-known Casimir elements of $\Sl_2(\K)$, a degree two
polynomial~\cite[\S 16.5]{SnobWinternitz2017book}, and of $\mathfrak{h}_6$, a
degree three polynomial~\cite{BallesterosHerranz2001}, which are recovered as
the $n=2,3$ cases, respectively. Then, in the spirit
of~\cite{Ballesteros_et_al_1996,BallesterosRagnisco1998,Ballesteros_et_al2009}
we build a hierarchy of Hamiltonian systems in $N$ degrees of freedom
associated to the Lie algebra $\galg_n$. In particular, we show that each of
these systems possesses a total number of $2(N-n+1)$ first integrals divided
into two families of $N-(n-2)$ pairwise commuting ones. In particular, this
implies that the given systems are integrable for $n=2$, quasi-integrable for
$n= 3$, and of Poincar\'e--Lyapunov--Nekhoroshev type with rank
$N-(n-2)$~\cite{Nekhoroshev1994,Gaeta2002,Gaeta2003} for $n\geq 4$.

\

The plan of the paper is the following: in \Cref{sec:background} we discuss the
background material we need throughout the paper. In particular, we focus on
the basic notions on Lie algebras and their representations that we are going
to use, and we give a short introduction to the coalgebra symmetry method
applied to Classical Mechanics. In \Cref{sec:gndef} we introduce the Lie
algebra $\galg_n$ and establish its main properties.  In particular, we show
that there is a chain of inclusions of Lie subalgebras, prove that $\galg_{n}$ can
have at most one non-trivial Casimir element, and then construct a matrix
representation in $\Sl_{2(n-1)}(\K)$ and also its coadjoint representation in
the space of vector fields. \Cref{sec:casgn} contains the main result of the
paper: the proof of the form of the Casimir element of $\galg_n$ for arbitrary
$n$. We start by employing the method of vector fields to examine the
low-dimensional examples $n=2,3,4$. Carefully analising those
examples we guess the general form of the Casimir polynomial for generic $n$ as
the determinant of a given $n\times n$ symmetric matrix with coefficients in
$\galg_{n}$. Then, we use the matrix representation constructed
in \Cref{sec:gndef} to prove that such a determinant is indeed a Casimir
polynomial, being well-behaved with respect to the action of the group
$\SL_n(\K)$, and recalling that such an action has the property of leaving
invariant the determinant of a symmetric matrix.  Then, in \Cref{sec:coalg} we
proceed to associate to $\galg_n$ a hierarchy of classical Hamiltonian systems
using a one degree of freedom symplectic
realisation.  Moreover, using the information on the Casimir element of
$\galg_n$, we are able to describe the first integrals of the obtained
Hamiltonian systems. In particular, we provide a expression of the first
integrals as sum of squares of ``building blocks'' that are $n$-indices objects
arising as linear combinations of the angular momentum components $L_{i,j}$,
plus constant terms coming from the $\galg_n$ central elements under the chosen
symplectic realization. We also provide a determinantal formula for these
$n$-indices objects, which clarifies their rather intricate structure. Of
course, the known cases $n=2,3$ are properly recovered.  Finally, in
\Cref{sec:concl} we discuss our results and suggest further directions of
research in the field.

\section{Background and notations}
\label{sec:background}

\noindent In this section we discuss the background tools we need to fully
develop our arguments, and fix the notations we are going to use throughout the
manuscript.

\subsection{Notions from representation theory}

Consider a finite-dimensional Lie algebra $\galg$ over a field $\mathbb{K}$.
In this subsection we recall a few notions from representation theory, see for
instance~\cite{FultonHarris1991}.  A \emph{representation} $\rho$ of
the Lie algebra $\galg$ on a vector space $V$ (finite or infinite dimensional)
is a linear map of $\galg$ into the space of linear
operators acting on $V$, denoted by $\mathcal{L}(V)$:
\begin{equation}
    \begin{tikzcd}[row sep=tiny]
        \rho \colon\galg \arrow{r} & \mathcal{L}(V)
        \\
        \phantom{\rho \colon}x \arrow[mapsto]{r} & \rho(x),
    \end{tikzcd}
    \label{eq:reprgen}
\end{equation}
such that:
\begin{equation}
    \rho([x,y]) = \rho(x)\circ\rho(y) - \rho(y)\circ\rho(x).
    \label{eq:reprcond}
\end{equation}
Note that, within this definition, we implicitly used that $\mathcal{L}(V)$ is
a Lie algebra with Lie bracket given by the difference of compositions:
${[F,G]}_{\mathcal{L}(V)}=F\circ G - G\circ F$, and a representation preserves
the Lie bracket: $\rho([x,y]) = [\rho(x),\rho(y)]$.  If the map $\rho$ is
injective, the representation is called \emph{faithful}.

When $V$ is a finite-dimensional vector space over $\K$ we have that
$\mathcal{L}(V)\cong \gl(V)$, the space of linear operators on $V$. If $\dim V
= n$, after choosing a basis, this is the space of $n\times n$ matrices and it
is denoted by $\gl_{n}(\K)$.  The commutator is $[X,Y]=XY-YX$,
cf.~\eqref{eq:reprcond}. Representations $\rho\colon\lieg\longrightarrow
\gl_n(\K)$ are called \emph{matrix representations}.

A basic example of a representation of a Lie algebra, is the \emph{adjoint
representation $\rho_{\adj}$}, which we will use later. 
It is defined as follows:
\begin{equation}
    \begin{tikzcd}[row sep=tiny]
        \rho_{\adj} \colon \galg \arrow{r} & \gl(\galg)
        \\
        \phantom{\rho_{\adj} \colon}x \arrow[mapsto]{r} &  \rho_{\adj}(x) (-)= \adj_{x}(-)=[x,-].
    \end{tikzcd}
    \label{eq:rhoadj}
\end{equation}
The Lie bracket is preserved due to the Jacobi identity for a Lie algebra.
Given a finite-dimensional representation $\rho\colon \galg \longrightarrow
\gl(V)$ we can define representations on tensor products of $V$, see~\cite[\S
4.2.2]{KirillovBook}.
In particular, we are interested in the
representation $\rho^{\otimes 2}$ induced on $V^{\otimes 2}$, namely
$\rho^{\otimes 2} \colon \galg \longrightarrow \gl(V^{\otimes 2})$, where for
every $x\in\galg$ we have:
\begin{equation}
    \begin{tikzcd}[row sep=tiny]
        \rho^{\otimes 2}(x)\colon V^{\otimes2}\arrow[r] & V^{\otimes 2}
        \\
        \phantom{\rho^{\otimes 2}(x)\colon} v\otimes w \arrow[r,mapsto] & (\rho(x)v)\otimes w\,+\,v\otimes(\rho(x)w)~.
    \end{tikzcd}
    \label{eq:rhot2abstract}
\end{equation}
We observe that that $\rho(x)$ acts as a derivation on $V^{\otimes 2}$ and this is
also true for $V^{\otimes n}$.  This action can be made more explicit by
choosing a basis $e_1,\ldots,e_{n}$ of $V$ and identifying the tensor product
with the vector space of $n\times n$ matrices in the following way:
\begin{equation}
    \begin{tikzcd}[row sep=tiny]
        V^{\otimes 2} \arrow[r,leftrightarrow,"\cong"] & M_{n}\left( \K \right)
        \\
        e_i\otimes e_j \arrow[r,leftrightarrow] & E_{i,j}~,
    \end{tikzcd}
    \label{eq:Vo2Mat}
\end{equation}
where $E_{i,j}$ is an $n\times n$ matrix whose
only non-zero entry is the element $(i,j)$:
\begin{equation}
    E_{i,j} = (\delta_{ik}\delta_{jl})_{k,l=1}^{n}.
    \label{eq:Eij}
\end{equation}
It is then an easy computation to verify that the representation $\rho^{\otimes 2}$
on $V^{\otimes 2}$ is (equivalent to) the representation, still denoted by
$\rho^{\otimes 2}$, on $M_n(\K)$:
\begin{equation}
    \begin{tikzcd}[row sep=tiny]
        \rho^{\otimes 2} \colon \galg \arrow{r} & M_n(\K)
        \\
        \phantom{\rho^{\otimes 2} \colon}x \arrow[mapsto]{r} &  \rho^{\otimes 2}(x)(M)
        =
        \rho(x)M + M{\rho(x)}^{T} .
    \end{tikzcd}
    \label{eq:rhot2}
\end{equation}
In the decomposition $V^{\otimes 2}= \Sym^{2} V \oplus \Lambda^{2} V$, the
subspaces of symmetric and skew-symmetric tensors are mapped into  themselves
by any $\rho^{\otimes 2}(x)$. So, it makes sense to consider the
representations on the two invariant subsets. In particular, we denote the
representation $\rho^{\otimes 2}$ restricted to $\Sym^{2} V$ (resp.
$\Lambda^{2} V$) by $\rho^{(2)}$ (resp. $\rho^{[2]}$).

More generally, for each $n\geq 0$ we have the representations $\rho^{(n)}$ on
$\Sym^{n} V$.  This leads to consider the symmetric algebra $SV$ of $V$, an
infinite dimensional vector space
\begin{equation}
    SV\,:=\,\bigoplus_{n= 0}^{\infty}\Sym^{n} V~,
    \label{eq:SVgen}
\end{equation}
with the convention that $\Sym^{0} V=\K$. An element $v\in V$ defines a linear
map on the dual vector space $v\colon V^*\longrightarrow \K$ simply mapping
$\ell \in V^*$ to
$\ell(v)$. The vector space  $\Sym^{n} V$ can then be identified with linear
combinations of products $v_{i_1}v_{i_2}\ldots v_{i_n}$ of functions on $V^*$
that are homogeneous of degree $n$. This leads to the identification
$SV=\Pol(V^*)$, the algebra of polynomial functions on $V^*$.

We are particularly interested in the case that $V=\galg$ and
$\rho=\rho_{\adj}$, the adjoint representation on $V$, now extended, by
derivations, to all of $SV$. That is,
$\rho_{\adj}\colon\lieg\,\longrightarrow\,\mathcal{L}(S\lieg)$, where for all
$x\in \galg$ we have:
\begin{equation}
    \begin{tikzcd}[row sep=tiny]
        \rho_{\adj}(x) \colon S\galg \arrow{r} & S\galg
        \\
        \phantom{\rho_{\adj}(x) \colon} x_{i_1}x_{i_2}\ldots x_{i_n} \arrow[mapsto]{r} &  
        {[x,x_{i_1}]}x_{i_2}\ldots x_{i_n}+
        \ldots+\,[x,x_{i_n}]x_{i_1}\ldots x_{i_{n-1}}.
    \end{tikzcd}
    \label{eq:rhoaddsymm}
\end{equation}
Then, a \emph{polynomial Casimir element} $C\in S\lieg$ is by definition an
element annihilated by the adjoint action:
\begin{equation}
    \rho_{\adj}(x)(C)\,=\,0\qquad \forall x\in\galg~.
    \label{eq:casdef}
\end{equation}
Clearly, the elements of the centre of a Lie algebra $Z(\galg)$ are Casimir
elements. These elements are usually called \emph{trivial Casimir elements},
since they belong to the Lie algebra $\galg$ itself.

More explicitly, if we choose a basis $\galg=\Span_{\K}\Set{x_{1},\ldots,x_{{d}}}$,
then we have the identification $S\galg \cong \K\left[ x_{1},\ldots,x_{{d}}
\right]$, the polynomial algebra in $d$ variables. The Lie algebra
structure is determined by the structure constants $c^k_{ij}$, i.e.:
\begin{equation}
    [x_i,x_j]\,:=\,\sum_{k=1}^{{d}}c^k_{ij}x_k~.
    \label{eq:cijk}
\end{equation}
The Lie bracket of $\galg$ induces on $S\galg=\Pol(\galg^{*})$, and more in
general on the algebra of functions ${\mathcal F}(\galg^*)$ of all
smooth/analytic functions on $\galg^*$, the so-called Lie--Poisson bracket:
\begin{equation}
    \left\{ f,g \right\}_{S\galg} =
    \sum_{i,j,k=1}^{{d}} c_{ij}^{k}x_{k}
    \frac{\partial f}{\partial x_{i}}
    \frac{\partial g}{\partial x_{j}},
    \label{eq:plbracket}
\end{equation}
see also~\cite[Chap. 7]{LaurentGengoux2018}.
Notice that:
\begin{equation}
    \rho_{\adj}(x_i)(x_j)\,=\,[x_i,x_j]\,=\,\{x_i,x_j\}_{{S\galg}},
    \label{eq:rhoadxixj}
\end{equation}
and hence, since $x\in \galg$ acts by derivations:
\begin{equation}
    \rho_{\adj}(x)(f)\,=\,\{x,f\}_{{S\galg}}.
    \label{eq:rhoadf}
\end{equation}
Consequently, polynomial Casimir elements $C$ can be constructed as elements of
the symmetric algebra by solving:
\begin{equation}
    \pb{C}{x}_{S\galg} = 0,
    \quad \forall x\in \galg.
    \label{eq:plCas}
\end{equation}

Since $\galg$ acts on ${\mathcal F}(\galg^*)$ by derivations, we actually
have a map $\galg\,\longrightarrow\,\mathfrak{X}(\galg^*)$, the latter being
the space of vector fields over $\galg^*$. In fact, the formula for
$\rho_{\adj}(x_i)(x_j)$ implies that $x_i$ acts as the vector field:
\begin{equation}
    \widehat{x}_i\,:=\,
    \sum_{j,k=1}^{{d}} c_{ij}^{k} x_{k} \frac{\partial}{\partial x_{j}}.
    \label{eq:eihat}
\end{equation}
The vector fields are a Lie algebra with Lie bracket given by
\begin{equation}
    [X,Y]_{LB}(f)\,:=\,X(Y(f))\,-\,Y(X(f)),
    \qquad 
    X,Y\in\mathfrak{X}(\galg^*),
    \quad
    f\in \mathcal{F}(\galg^*).
    \label{eq:LBvf}
\end{equation}
Then $x\mapsto\widehat{x}$ is a Lie algebra representation
since the vector fields $\widehat{x}_i$
satisfy the same commutation relations as $x_i$ in the Lie algebra $\galg$,
i.e.\ the following identity holds:
\begin{equation}
    {[\widehat{x}_{i},\widehat{x}_{j}]}_{LB}
    =\widehat{[x_{i},x_{j}]}_{\galg} \, .
    \label{eq:javv}
\end{equation}
Thus, following for instance~\cite{SnobWinternitz2017book}, we say
that the vector fields $\Set{\widehat{x}_{1},\ldots,\widehat{x}_{{d}}}\subset
\mathfrak{X}(\galg^*)$ form a basis of the so-called coadjoint representation
of $\galg$, acting by derivations on $\mathcal{F}(\galg^*)$.

\begin{remark}
    Note that if $Z(\galg)\neq 0$ the adjoint representation is not
    faithful, as all the elements in the centre are mapped into the zero
    vector field of $\mathfrak{X}(\galg^*)$.
    \label{rem:Zentrum}
\end{remark}

For the search of Casimirs, we then have the following result:

\begin{proposition}[{\cite[\S 3.2]{SnobWinternitz2017book}}]
    Given a $d$-dimensional Lie algebra $\galg$, then its Casimir elements can be
    found as the functionally independent solutions of
    the system of PDEs:
    \begin{equation}
        \widehat{x}_{i}( f ) = 0, \quad
        i=1,\ldots,d,
        \label{eq:pdej}
    \end{equation}
    where the $\widehat{x}_{i}$ are the basis of the coadjoint representation in
    the space of vector fields, see equation~\eqref{eq:eihat}.
    \label{prop:casj}
\end{proposition}

This result tells us that one does not need to check the polynomial Casimir
condition~\eqref{eq:plCas} for every element, but just a suitably chosen basis
is enough.  For applications and further discussion of this method to find
Casimir elements of Lie algebras we refer to~\cite{SnobWinternitz2017book}. A
short discussion of this approach in our case of interest, along
with some motivations on why it tends to become ineffective as the dimension of
the Lie algebra grows, will be presented in \Cref{sec:casgn}.

We also recall the following important result on the number of Casimir elements
admitted by a Lie algebra:

\begin{theorem}[Beltrametti--Blasi~\cite{BeltramettiBlasi1966}]
    Let $\galg$ be a $d$-dimensional Lie algebra over $\K$ with structure constants
    $c_{ij}^{k}$.  Then, the number of functionally independent Casimir
    invariants, $\nu$, of $\galg$ is the following:
    \begin{equation}
        \nu = \dim \galg
        - \rank A,
        \qquad
    A = {\left(\sum_{k=1}^{{d}} c_{ij}^{k}x_{k}  \right)}_{i,j=1}^{{d}},
        \label{eq:bbformula}
    \end{equation}
    where the $x_{i}$ are commuting variables on the base field.
    \label{thm:beltramettiblasi}
\end{theorem}

We conclude by recalling that for some other authors Casimir elements are
considered in the universal enveloping algebra (UEA) $U\galg$ of the Lie
algebra $\galg$ rather than in its symmetric algebra. For instance,
in~\cite{KirillovBook} a Casimir element is defined to be an element of
$U\galg$ invariant with respect to the adjoint action of $\galg$, defined by
extending the adjoint action of $\galg$ on $U\mathfrak{g}$. In fact, this
definition is analogous to the one we gave, because there exists a map
$\Sigma\colon S\galg \longrightarrow U\galg$, called the \emph{symmetrisation
map}, mapping polynomial Casimir elements in $S\lieg$ into polynomial Casimir
elements in $U\lieg$ in the sense just described.  For this reason, within this
paper we will only consider Casimir elements in the symmetric algebra $S\galg$,
see also~\cite[\S 2]{Campoamoretal2023}.

\subsection{The coalgebra symmetry approach for Hamiltonian systems}
\label{sss:coalg}

The concept of coalgebra was introduced in quantum group theory during the
1980s~\cite{Drinfeld1987}, see also the book~\cite{ChariPressley1994Book}.
However, its application to (super)integrable systems did not emerge until
the late 1990s through the pioneering work of Ballesteros, Ragnisco, and
their collaborators
in~\cite{Ballesteros_et_al_1996,BallesterosRagnisco1998}. Over the years,
this approach has been significantly expanded to encompass diverse
scenarios, leading to a deeper understanding of various integrable and
superintegrable systems. As evident from comprehensive reviews such
as~\cite{Ballesteros_et_al2009}, the coalgebra symmetry framework has also
been successfully applied in other areas, including the study of
superintegrable systems on non-Euclidean
spaces~\cite{BallesterosHerranz2007,Ballesteros_et_al2008PhysD}.
Additionally, it has been employed to investigate models with spin-orbital
interactions~\cite{Riglioni_et_al2014}, discrete quantum mechanical
systems~\cite{LatiniRiglioni2016}, and superintegrable systems related to
the generalized Racah algebra $R(n)$~\cite{DeBie_et_al2021, Latini2019,LMZ2021,LMZ2021b}.
Furthermore, adaptations of this method have been made for use with
comodule algebras in~\cite{Ballesteros_et_al2002} and loop coproducts in
\cite{Musso2010loop}. Recent work has also explored its application to
discrete-time systems\cite{GubLat_sl2,GLT_coalgebra,DrozGub_h6}. This
subsection will provide a brief overview of the coalgebra symmetry
construction, highlighting the main results and ideas in this theory.

The starting point of this approach is the following definition:

\begin{definition}
    A \emph{coalgebra} is a pair $(\mathcal{A},\Delta)$ where $\mathcal{A}$ is
    a unital, associative algebra and $\Delta\colon \mathcal{A} \longrightarrow
    \mathcal{A}^{\otimes2}$ is a \emph{coassociative} map.  That is,
    $\Delta$ is an algebra homomorphism from $\mathcal{A}$ to
    $\mathcal{A}^{\otimes2}$:
    \begin{equation}
        \Delta (x \cdot  y) = \Delta (x)  \cdot \Delta (y), \qquad \forall x, y \in \mathcal{A} \, ,
        \label{eq:homalg}
    \end{equation}
    such that:
    \begin{equation}
        (\Delta \otimes \id) \circ \Delta=(\id \otimes \Delta) \circ \Delta \, ,
        \label{eq:commdiag}
    \end{equation}
    that is the following diagram commute:
    \begin{equation}
        \begin{tikzpicture}[baseline={(0,0)},thick]
            \node (a1) at (0,0){$\mathcal{A}$};
            \node (a2) at ($({2*cos(30)},{2*sin(30)})$)
            {$\mathcal{A}^{\otimes2}$};
            \node (a3) at ($({2*cos(30)},{-2*sin(30)})$) {$\mathcal{A}^{\otimes2}$};
            \node (a4) at (4,0){$\mathcal{A}^{\otimes3}$};
            \draw[->] (a1) edge node[above left]{$\Delta$} (a2) (a2) edge node[above right]{$\Delta\otimes\id$}(a4) ;
            \draw[->] (a1) edge node[below left]{$\Delta$} (a3) (a3) edge node[below right]{$\id\otimes\Delta$} (a4);
        \end{tikzpicture}
        \label{eq:commdiag2}
    \end{equation}
    The map $\Delta$ is called the \emph{coproduct map}. 
    \label{def:coalgebra}
\end{definition}

When there is no possible confusion on the coproduct map, it is customary
to denote the coalgebra simply by $\mathcal{A}$. Our case of interest are
\emph{Poisson coalgebras}. That is, we will not assume $\mathcal{A}$ to be
a generic unital associative algebra, but rather we will consider it to be a
\emph{Poisson algebra}, i.e.\ a commutative associative algebra admitting a
bilinear derivation $\left\{ \, , \, \right\}$, satisfying the Jacobi
identity. The bilinear operation $\left\{ \,,\, \right\}$ is called a
\emph{Poisson bracket}. To build a Poisson coalgebra, the coalgebra structure
and the Poisson algebra structure have to be compatible in
the sense of the following definition:

\begin{definition}
    A \emph{Poisson coalgebra} is a triple $(\mathcal{A},\Delta,\left\{ \,
    , \, \right\})$ where $(\mathcal{A},\Delta)$ is a coalgebra and
    $(\mathcal{A},\left\{ \, , \, \right\})$ is a Poisson algebra, such
    that the coproduct map is a \emph{Poisson homomorphism} with respect to
    the induced Poisson bracket on $\mathcal{A}^{\otimes2}$:
    \begin{equation}
        \left\{x \otimes y, w \otimes z\right\}_{\mathcal{A}^{\otimes2}}
        =
        \left\{x, w\right\}_{\mathcal{A}} \otimes yz
        +xw \otimes \left\{y,z\right\}_{\mathcal{A}}, 
        \quad 
        x,y,z,w \in \mathcal{A},
        \label{eq:pstrucatensa}
    \end{equation}
    i.e.\ $\Delta(\left\{x,y\right\}_{\mathcal{A}})=\left\{
    \Delta(x),\Delta(y) \right\}_{\mathcal{A}^{\otimes2}}$.
    \label{def:poissoncoalgebra}
\end{definition}

The coproduct provides us a way to obtain
relevant elements living in the two-fold tensor product of a Poisson
algebra. In particular, the coproduct is not surjective, but its image is a
copy of the original Poisson algebra $\mathcal{A}$ in $\mathcal{A}^{\otimes
2}$: $\mathcal{A} \cong\Delta(\mathcal{A}) \subset \mathcal{A}^{\otimes2}$.

The construction of the coproduct can be extended in two different
ways for all $m>2$ \cite{Ballesteros_et_al2009}. Indeed, by recursion one
can define the \emph{left and right} higher-order $m$th \emph{coproducts}
$\Delta^{(m)}_{R}\colon  \mathcal{A}\longrightarrow \mathcal{A}^{\otimes m}$ and
$\Delta^{(m)}_{L}\colon \mathcal{A} \longrightarrow \mathcal{A}^{\otimes m}$, acting as:
\begin{subequations}
    \begin{align}
        \Delta_{L}^{(m)} &=
        \bigl(\overbrace{\id\otimes \id\otimes \dots \otimes \id}^{m-2}\otimes
        \Delta\bigr)\circ\Delta^{(m-1)}_{L},
        \label{left-higher-order-coproducts}
        \\
        \Delta^{(m)}_{R} &=
        \bigl(\Delta\otimes\overbrace{\id\otimes \id\otimes \dots \otimes \id}^{m-2}\bigr)\circ\Delta^{(m-1)}_{R}.
        \label{right-higher-order-coproducts}
    \end{align}
    \label{higher-order-coproducts}%
\end{subequations}
By induction on $m$, it can be proved that the $m$th coproducts
$\Delta^{(m)}_{L}$ and $\Delta_{R}^{(m)}$ are also Poisson maps on
$\mathcal{A}^{\otimes m}$ and define a copy of the original algebra:
    $\mathcal{A} \cong\Delta^{(m)}_{i}(\mathcal{A}) \subset \mathcal{A}^{\otimes m}$,
    $i=L,R$.

In particular, if $\Set{x_{1},\ldots,x_{{d}}}$ is the generating set of
$\mathcal{A}$, for a generic element $F\in\mathcal{A}$ we can write:
\begin{subequations}
    \begin{align}
        \Delta_{L}^{(m)}({F})(x_{1},\ldots,x_{{d}}) &= 
        {F}(\Delta_{L}^{(m)}(x_{1}),\ldots,\Delta_{L}^{(m)}(x_{{d}})),
        \label{eq:copfunctL}
        \\
        \Delta_{R}^{(m)}({F})(x_{1},\ldots,x_{{d}}) &= 
        {F}(\Delta_{R}^{(m)}(x_{1}),\ldots,\Delta_{R}^{(m)}(x_{{d}})).
        \label{eq:copfunctR}
    \end{align}
    \label{eq:copfunct}
\end{subequations}

Now, let us fix an $N\in\N$.  Then, following for
example~\cite{Musso2010gaudin} we can consider a chain of tensor products
$\mathcal{A}^{\otimes m}$ for all $m \leq N$ with left/right embeddings:
\begin{equation}
    \begin{tikzcd}[row sep=tiny]
        \iota_{m,L} \colon \mathcal{A}^{\otimes m} \arrow{r} & \mathcal{A}^{\otimes N}
        \\ 
        \phantom{\iota_{m,L}} x \arrow[mapsto]{r} & \iota_{m,L}(x)
        =
        x \otimes \underbrace{1\otimes 1\otimes \dots \otimes 1}_{N-m}, 
    \end{tikzcd}
    \label{eq:Algembed}
\end{equation}
and:
\begin{equation}
    \begin{tikzcd}[row sep=tiny]
        \iota_{m,R} \colon \mathcal{A}^{\otimes {m}} \arrow{r} & \mathcal{A}^{\otimes N}
        \\ 
        \phantom{\iota_{m,R}} x \arrow[mapsto]{r} & \iota_{m,R}(x)
        =
        \underbrace{1\otimes 1\otimes \dots \otimes 1}_{{N-m}}\otimes x. 
    \end{tikzcd}
    \label{eq:Algembed2}
\end{equation}
So, we can consider all the elements as lying in the final tensor
product $\mathcal{A}^{\otimes N}$.  In particular, we can consider all the
Poisson brackets in $\mathcal{A}^{\otimes N}$ in the following way: let
$x\in\mathcal{A}^{\otimes m}$ and $y\in\mathcal{A}^{\otimes m'}$, then:
\begin{equation}
    \left\{ x, y \right\}_{\mathcal{A}^{\otimes N}} = 
    \left\{ \iota_{m,L}(x),\iota_{m',L}(y) \right\}_{\mathcal{A}^{\otimes N}},
    \label{eq:algNpoisson}
\end{equation}
and, similarly, if we consider the tensor product within the right embedding,
or the mixed case.

Now, a crucial ingredient in the construction of Hamiltonian systems from
coalgebras is the fact that Poisson algebras admit Casimir
elements. This is a direct generalisation of what we discussed in the previous
section about Casimir elements for Lie algebras. Clearly, given a
Poisson algebra $(\mathcal{A},\left\{ .,. \right\})$, a
Casimir element $C\in\mathcal{A}$ is an element such that $\left\{ C,x
\right\}_{{\mathcal{A}}}=0$ for all $x \in\mathcal{A}$, to be compared with
\eqref{eq:plCas}.  We denote the space of Casimir elements of a
Poisson algebra by:
\begin{equation}
    \Cas(\mathcal{A},\left\{ .,. \right\}) =
    \Set{ C \in \mathcal{A} | \left\{ C,x \right\}=0, \,\forall x\in\mathcal{A} }.
    \label{eq:casimirspace}
\end{equation}
This space is a subalgebra and also an ideal of $\mathcal{A}$~\cite[Chap.\
1]{LaurentGengoux2018}. If $\Cas(\mathcal{A},\left\{ .,.  \right\})$ is
\emph{finitely generated} we can take a (minimal) set of its generators
$\Set{C_{1},\dots,C_{s}}$. We will say that such a set of generators is a set
of \emph{functionally independent Casimirs}. This name is justified from the
fact that in the case of the symmetric algebra $S\galg$ of a $d$-dimensional
Lie algebra $\galg$, through the isomorphism $S\galg\cong
\K[x_{1},\ldots,x_{d}]$ this notion reduces to the standard analytic notion of
functional independence.

Then, we can state the following fundamental result adapted to our notations:

\begin{theorem}[\cite{Ballesteros_et_al_1996,Ballesteros_et_al2009}]
    Assume we are given a finitely generated Poisson
    coalgebra $(\mathcal{A}, \Delta, \left\{ \,,\, \right\})$, with
    $\mathcal{A}$ generated by
    the set $\{x_1, \dots, x_d\}$, whose Casimir subalgebra is finitely
    generated $\Cas(\mathcal{A},\left\{ \,,\, \right\}) =
    \Set{C_{1},\ldots,C_{s}}$. Then, the two sets of functions:
    \begin{subequations}
        \begin{align}
          \rm{L} &= \Set{
            C_{j}^{(m)} :=
        (\iota_{m,L}\circ\Delta^{(m)}_{L})(C_j)(x_1,\dots,x_{d})}_{m=1,\dots,N,j=1,\dots,s},
            \label{eq:Ctotg}
            \\
            \rm{R} &= \Set{
            C_{j, (m)} := (\iota_{m,R}\circ\Delta^{(m)}_{R})(C_j)(x_1,\dots,x_{d})}_{m=1,\dots,N,j=1,\dots,s},
            \label{eq:rfun}
        \end{align}
        \label{eq:rlfunc}%
    \end{subequations}
    are made of Poisson-commuting functions on $\mathcal{A}^{\otimes N}$.
    Furthermore, they Poisson commute with
    $\Delta^{(N)}_{L}(x_i)=\Delta^{(N)}_{R}(x_i)$, $i=1, \dots, d$ and with all
    functions $\Delta_{L}^{(N)}(H)=\Delta_{R}^{(N)}(H)\in \mathcal{A}^{\otimes
    N}$, where $H \in \mathcal{A}$.
    \label{thm:fundcoalg}
\end{theorem}

This result was obtained in~\cite{Ballesteros_et_al_1996} in the case of the
left coproduct. The extension to incorporate also the right coproducts
originally introduced in~\cite{Ballesteros_et_al_2004} can be found, for
example, in~\cite{Ballesteros_et_al2009}.  The functions $C_{j}^{(m)}$ and
$C_{j,(m)}$ are usually called \emph{left and right Casimir functions}. Note
that, despite the name these functions are not Casimir elements in the standard
sense, and in general they do not commute with each other~\cite{Latini2019}.
Moreover, because of the coassociativity property we have
$C_{j}^{(N)}=C_{j,(N)}$, at a fixed $j=1, \ldots s$.

Now, to discuss integrable systems, we need to put this algebraic construction
with a Hamiltonian system. This is done through the so-called symplectic
realisation of a Poisson algebra. A symplectic realisation
is, in general, a Poisson algebra homomorphism from a given Poisson algebra
$(\mathcal{A},\left\{\,,\, \right\})$ to the algebra of smooth functions on an
open subset $U$ of a symplectic manifold $(\mathcal{M},\Omega)$, with $\dim
\mathcal{M}=2M$. Since we are working locally, from Darboux theorem, we
can consider on the open subset $U$ the canonical coordinates
$(\vb{q}_{1},\vb{p}_{1})$, and the symplectic form to be the canonical one
$\omega$. Under such hypotheses a \emph{$M$ degrees of freedom canonical
symplectic realisation } is the following Poisson homomorphism:
\begin{equation}
    \begin{tikzcd}[row sep=tiny]
        \DD\colon (\mathcal{A},\left\{ \,,\, \right\})  \arrow{r} & 
        (\mathcal{C}^{\infty}(U),\left\{\, ,\, \right\}_{\omega})
        \\ 
        \phantom{\mathrm{D}\colon S\galg}x_{i} \arrow[mapsto]{r} &
        x_{i}(\vb{q}_{1},\vb{p}_{1}).
    \end{tikzcd}
    \label{eq:grepr}
\end{equation}

If the Poisson algebra is a Poisson coalgebra, from a $M$-degrees of freedom
symplectic realisation we can find a $N M$-degrees of freedom symplectic
realisation by taking the tensor product of realisations: $\DD^{\otimes N}$. The
factors of the tensor product are usually physically interpreted as ``sites'',
where on each of them resides a ``particle'' with $M$-degrees of freedom. 
We denote the canonical coordinates in the $NM$-degree of freedom
realisation by:
\begin{subequations}
\begin{equation}
 (\vb{q},\vb{p}) = (\vb{q}_{1},\vb{p}_{1},\ldots,\vb{q}_{N},\vb{p}_{N}),
    \label{eq:genqpa}
\end{equation}
where:
\begin{equation}
    \vb{q}_{j} ={(q_{j,1},\ldots,q_{j,M})}^{T},\quad
\vb{p}_{j}={(p_{j,1},\ldots,p_{j,M})}^{T},
    \qquad j=1,\ldots,N.
    \label{eq:genqpb}
\end{equation}
    \label{eq:genqp}
\end{subequations}

In general, by applying the tensor product realisation we can construct the
first integrals, using the left and right Casimir introduced
in~\Cref{thm:fundcoalg}. Moreover, using the last part of~\Cref{thm:fundcoalg}
we see that we can lift a one-particle, $M$-degrees of freedom Hamiltonian $H
\in \mathcal{C}^{\infty}(U)$ to generate an $N$ particle, $NM$-degrees of
freedom Hamiltonian equipped with \emph{many} first integrals, usually called
\emph{universal integrals}~\cite{BallesterosHerranz2007,Latini2019}. This is
the content of the following result:

\begin{theorem}
    Assume we are given a finitely generated Poisson coalgebra $(\mathcal{A},
    \Delta, \left\{ \,,\, \right\})$ whose Casimir subalgebra is finitely
    generated $\Cas(\mathcal{A},\left\{ \,,\, \right\}) =
    \Set{C_{1},\ldots,C_{s}}$, and assume we are given a $M$-degrees of
    freedom canonical realisation $D\colon\mathcal{A}\longrightarrow
    \mathcal{C}^{\infty}(U)$. Then, given $H\in \mathcal{A}$ and the corresponding
    $M$-degrees of freedom Hamiltonian $H(\vb{q}_{1},\vb{p}_{1}):=\DD(H)\in
    \mathcal{C}^{\infty}(U)$, the $N M$-degrees of freedom Hamiltonian:
    \begin{equation}
        H^{(N)}(\vb{q},\vb{p}):= (\DD^{\otimes N}\circ \Delta_{L}^{(N)})(H) 
        \label{eq:HN}
    \end{equation}
    admits two sets of $sN$ commuting first integrals given by:
    \begin{subequations}
        \begin{align}
          L &= 
            \Set{C_{j}^{(m)}(\vb{q},\vb{p}):= \DD^{\otimes N}(C_{j}^{(m)}) }_{m=1,\dots,N,j=1,\dots,s},
            \label{eq:CLRrealF}
            \\
            R &= 
            \Set{C_{j,(m)} (\vb{q},\vb{p}):= \DD^{\otimes N}(C_{j,(m)}) }_{m=1,\dots,N,j=1,\dots,s},
            \label{eq:CLRrealG}
        \end{align}
        \label{eq:CLRreal}
    \end{subequations}
    where the functions $C_{j}^{(m)}$ and $C_{j, (m)}$ are defined in equation~\eqref{eq:rlfunc}.
    \label{thm:coalgfund2}
\end{theorem}

Let us now discuss the main case of interest, that is the one of Lie--Poisson
algebras $(\mathcal{F}(\galg^*),\{,\})$, or if only polynomials are involved
the symmetric algebra $S\galg \cong \Pol(\galg^*)\subset \mathcal{F}(\galg^*)$
of a Lie algebra $\galg$.  All Lie--Poisson algebras admit a common coproduct,
called the \emph{primitive coproduct} $\Delta$~\cite{Tjin1992}, which can be
defined on the (linear) generators of the symmetric algebra as:
 \begin{equation}
    \begin{tikzcd}[row sep=tiny]
        \Delta \colon S\galg \arrow{r} & (S\galg)^{\otimes2}
        \\ 
        \phantom{\Delta} x_{i} \arrow[mapsto]{r} & \Delta(x_{i})
        = x_{i} \otimes 1 + 1 \otimes x_{i}, 
    \end{tikzcd}
    \label{eq:DeltaP}
\end{equation}
then extended to generic elements of $S\galg$ through the homorphism
property~\eqref{eq:homalg}, and eventually to the whole algebra of smooth
functions $\mathcal{F}(\galg^{*})$. Many different coproducts are build upon
deformation of the primitive coproduct, see e.g.~\cite{ChariPressley1994Book}
and the papers~\cite{LyakhovskyMudrov1992,BallesterosHerranz1996Rmat}.  The
primitive coproduct allows great simplifications in the theory.  For instance,
it makes the construction of the tensor realisations from a basic one
immediate. For the sake of convenience, we enclose this result in the following
lemma:

\begin{lemma}
    Consider a Lie algebra $\galg=\Span_{\R}\Set{e_{1},\ldots,e_{d}}$, the
    induced primitive coproduct on $S\galg\cong
    \R[x_{1},\ldots,x_{d}]$~\eqref{eq:DeltaP}, and a $M$-degrees of freedom
    realisation $\DD\colon (\mathcal{A},\pb{\;}{\;})  \longrightarrow
    (\mathcal{C}^{\infty}(U),\pb{\;}{\;}_{\omega})$~\eqref{eq:grepr}. Then, we
    have the following expressions for the $N M$-degrees of freedom canonical symplectic realisation
    of the generators: 
    \begin{subequations}
        \begin{align}
            \DD^{\otimes N}( \iota_{m,L}\circ \Delta_{L}^{(m)}(x_{i}) ) &=
            \sum_{j=1}^{m}x_{i}(\vb{q}_{j},\vb{p}_{j}),
            \label{eq:tensorXL}
            \\
            \DD^{\otimes N}(\iota_{m,R}\circ {\Delta_{R}^{(m)}}(x_{i}) ) &=
            \sum_{j=N-m+1}^{N}x_{i}(\vb{q}_{j},\vb{p}_{j}),
            \label{eq:tensorXR}
        \end{align}
        \label{eq:tensorX}
    \end{subequations}
    where $\vb{q}_{j}$ and $\vb{p}_{j}$ are given in equation~\eqref{eq:genqpb}.
    \label{lem:coprodreal}
\end{lemma}

\

Let us summarise the construction and place it in the context of
integrability of Hamiltonian systems. From \Cref{thm:coalgfund2} we have that
a Hamiltonian system endowed with coalgebra symmetry is naturally equipped with
two sets of first integrals whose number depends on the Casimir of the
underlying Lie (or Poisson) algebra. However, for a
``generic'' Poisson algebra such a Hamiltonian system is not
necessarily integrable, while for some specific ones it can be \emph{even more}
than integrable. To this end, we recall some basic definitions in the theory of
integrable systems:
\begin{itemize}
    \item a Hamiltonian system possessing $N$ functionally independent first
        integrals in involution, Hamiltonian included, is called
        \emph{Liouville integrable}~\cite{Arnold1997};
    \item a Liouville integrable system possessing $N+k$, $1\leq k\leq N-1$
        functionally independent first integrals is called \emph{superintegrable},
        in particular if $k=N-1$ it is called \emph{maximally superintegrable
        (MS)}, and if $k=N-2$ it is called \emph{quasi-maximally
        superintegrable (QMS)}~\cite{MillerPostWinternitz2013R,Rauch1983};
    \item a Hamiltonian system possessing $l$ with $1\leq l\leq N-1$ commuting
        first integrals is called \emph{a system of
        Poincar\'e--Lyapunov--Nekhoroshev (PLN) type with rank $l$}, in
        particular if $l=N-1$ it is called
        \emph{quasi-integrable}~\cite{Nekhoroshev1994,Gaeta2002,Gaeta2003}.
\end{itemize}
Now, for a Hamiltonian system with coalgebra symmetry there are algebraic
conditions to be fulfilled on the algebra and on its realisation to obtain
automatically Liouville integrability, see~\cite{Ballesteros_et_al2009}.
Moreover, for some specific realisation, the first integrals computed through
formulas~\eqref{eq:CLRreal} can be trivial up to a certain $N$, i.e.\ a
constant. Without entering into the details, we simply note that the only
aspect that may fail in terms of integrability concerns the functional
independence of the first integrals produced via \Cref{thm:coalgfund2}.
However, for any given Casimir, the left and right integrals generated from it
are, by construction, mutually functionally independent. We will show the
details of this construction in the explicit example of the Lie algebra
$\galg_{n}$ in \Cref{sec:coalg}.

\section{The $\mathfrak{g}_{n}$ chain of Lie algebras}
\label{sec:gndef}

\noindent In this section we rigorously introduce the $\mathfrak{g}_{n}$ Lie
algebra and derive its first properties. In particular, we present a matrix
representation of $\mathfrak{g}_{n}$ in the special linear algebra of a
properly chosen dimension.

Throughout this section and the rest of the paper we will not explicitly
indicate the base field for the Lie algebras. In general, the reader can
assume either the complex or the real numbers. Field-dependent results or
statements will be specified when needed.

\subsection{Definition and main properties}

The Lie algebra $\mathfrak{g}_n$ is a Lie algebra of dimension
$T_{n}=n(n+1)/2$ characterized by the following
generators\footnote{We use the notation $T_{n}$ because the
dimensions of the Lie algebras $\galg_{n}$ are the \emph{triangular numbers}.}:
\begin{itemize}
    \item the $3$ generators $\Set{h, x_-,x_+}$;
    \item the $2(n-2)$ generators $\Set{y_{i,-}, y_{i,+}}$ for $i=1, \dots,
        n-2$, ($n \geq 3$);
    \item the $T_{n-2}$ central elements $z_{i,j}$ with $1\leq i\leq j\leq n-2$,
        ($n \geq 3$);
\end{itemize}
with commutation relations:
\begin{equation}
    \begin{array}{lll}
        [x_+, x_-]=h, &  
        [h, x_\pm]=\pm 2 x_\pm,
        &
        [h, y_{i,\pm}]=\pm  y_{i,\pm}, 
        \\{}
        [x_-, y_{i,-}]=[x_+, y_{i,+}]=0,
        &
        [x_-, y_{i,+}] =y_{i,-}, & [x_+, y_{i,-}]=y_{i,+}, 
        \\{}
        [y_{i,+},y_{j,+}]= [y_{i,-},y_{j,-}]=0 &
        [y_{i,+}, y_{j,-}]=z_{\min(i,j), \max(i,j)}
        &
        [z_{i,j},-] = 0.
    \end{array}
    \label{eq:gncommrel}
\end{equation}
As stated above, the centre of $\galg_{n}$ is
\begin{equation}
    Z(\galg_{n}) = \Span_{\K}\Set{ z_{i,j}}_{1\leq i \leq j \leq n-2}.
    \label{eq:Zgn}
\end{equation}
Note that 
$\dim Z(\galg_{n})=T_{n-2} = (n-1)(n-2)/2$.

In the following we denote by $\overline{\galg}_{n}$ the
quotient algebra obtained as the quotient of $\galg_n$ by the ideal
$Z(\galg_n)$, i.e.\ $\overline{\galg}_{n}=\galg_{n}/Z(\galg_{n})$, see for
instance~\cite[\S 9.1]{FultonHarris1991}.

By construction, for $n=2$ we have $\mathfrak{g}_{2}=\Sl_{2}(\K)$.  Indeed, in
this case the generators $y_{i,\pm}$ and $z_{i,j}$ are not present.  Next, for
$n=3$ we get a six-dimensional Lie algebra with generators
$\Set{h,x_+,x_-,y_+:=y_{1,+}, y_-:=y_{1,-}, z:=z_{1,1}}$. This Lie algebra is
isomorphic to the two-photon Lie algebra
$\mathfrak{h}_{6}$~\cite{Zhang_et_al1990}, which is in turn isomorphic to the
$(1+1)$-Schr\"odinger algebra~\cite{Hagen1972,Niederer1972}. For $n=4$, we have a
ten-dimensional Lie algebra isomorphic to the $\mathfrak{A}_{10}$ Lie algebra
introduced in~\cite{GLT_coalgebra}. The peculiarity of these Lie algebras is
the fact that they form a chain of inclusions:
\begin{equation}
    \Sl_{2}(\K)\subset \mathfrak{h}_{6}\subset \mathfrak{A}_{10},
    \label{eq:chain}
\end{equation}
where each Lie algebra is ``embedded'' as a Lie subalgebra in the next one.
This result can be made more general in the following lemma:

\begin{lemma}
    For a fixed $n \, \geq 2$ we have the following chain of inclusions
    \begin{equation}
        \galg_{2}\subset \mathfrak{g}_{3}\subset \cdots \subset \mathfrak{g}_{n-1}
        \subset \galg_{n},
        \label{eq:chaing}
    \end{equation}
    where any Lie algebra is a Lie subalgebra of the following one.
    \label{lem:inclusion}
\end{lemma}

\begin{proof}
    Let us note that, given the isomorphisms discussed above, we can restrict
    to prove this statement for $n\geq 3$. To be more precise, we observe that 
    for all $2<k\leq n$ the following decomposition \emph{as vector spaces}
    holds true:
    \begin{equation}
        \galg_{k} = \galg_{k-1} \oplus
        I_{k} \supset \galg_{k-1},
        \label{eq:galgdecsub}
    \end{equation}
    where:
    \begin{equation}
        I_{k} = \Span_{\K}\Set{ y_{k-2,-},y_{k-2,+}}
        \oplus
        \Span_{\K}\Set{ z_{1,k-2},\ldots,z_{k-2,k-2}}.
        \label{eq:Ik}
    \end{equation}
    Now, we see that $[\galg_{k},I_{k}]\subset I_{k}$, that is $I_{k}$ is an
    ideal in $\galg_{k}$, and that $[\galg_{k-1},\galg_{k-1}]\subset
    \galg_{k-1}\subset \galg_k$, that is $\galg_{k-1}$ is a Lie subalgebra of
    $\galg_{k}$. Both these statements can be read from the following
    decomposition of the commutation table:
    \begin{equation}
        \begin{array}{c|c|cc|ccc}
            & \galg_{k-1} & y_{k-2,+} & y_{k-2,-} & z_{1,k-2} & \cdots & z_{k-2,k-2}
            \\
            \toprule
            \galg_{k-1} & L_{k-1} & \multicolumn{2}{c|}{W_{k}} & \vb{0} & \cdots & \vb{0}
            \\
            \midrule
            y_{k-2,+} & \multirow{2}{*}{$-W_{k}^{T}$}  & 0 & -z_{k-2,k-2} & 0 & \cdots & 0 
            \\
            y_{k-2,-} &  & z_{k-2,k-2} & 0 & 0 & \cdots & 0
            \\
            \midrule
            z_{1,k-2} & \vb{0}^T & 0 & 0 & 0 & \cdots &0
            \\
            \vdots & \vdots & \vdots & \vdots & \vdots & \ddots & \vdots
            \\
            z_{k-2,k-2} & \vb{0}^T & 0 & 0 & 0 & \cdots & 0
            \\
            \bottomrule
        \end{array}
        \label{eq:tabledec}
    \end{equation}
    where $L_{k-1}$ is the commutation table for $\galg_{k-1}$, and
    \begin{equation}
        W_{k} :=
        \begin{pmatrix}
            -y_{k-2,-} & y_{k-2,+}
            \\
            0 & y_{k-2,-}
            \\
            y_{k-2,+} & 0
            \\
            0 & -z_{1,k-2}
            \\
            z_{1,k-2} & 0
            \\
            0 & 0
            \\
            0 & -z_{2,k-2}
            \\
            z_{2,k-2} & 0
            \\
            0 & 0
            \\
            0 & 0
            \\
            0 & -z_{3,k-2}
            \\
            z_{3,k-2} & 0
            \\
            0 & 0
            \\
            \vdots & \vdots
            \\
            0 & 0
            \\
            0 & -z_{k-3,k-2}
            \\
            z_{k-3,k-2} & 0
        \end{pmatrix}.
        \label{eq:abdec}
    \end{equation}
    This ends the proof of the lemma.
\end{proof}

So, following the statement of the previous Lemma, we say that the
sequence of Lie algebras $\set{\galg_{n}}_{n \geq 2}$ form a \emph{chain of
Lie algebras}. Next, we show the Levi decomposition~\cite[\S
2.2]{SnobWinternitz2017book} of the Lie algebra $\galg_{n}$:

\begin{lemma}
    The Lie algebra $\galg_{n}$ has the following Levi decomposition:
    \begin{equation}
        \galg_{n} = \Sl_2(\K) \oplus_{S} \mathfrak{r}_{n},
        \label{eq:ld}
    \end{equation}
    where $\mathfrak{r}_{n}$ is the radical, generated by the elements
    $y_{i,\pm}$, $i=1,\ldots,n-2$ and $z_{i,j}$, $i\leq j = 1,\ldots,n-2$,
    whose non-zero commutation relations are:
    \begin{equation}
        [y_{i,+},y_{j,-}] = z_{\min(i,j),\max(i,j)}.
        \label{eq:genheinsenberg}
    \end{equation}
    \label{lem:ld}
\end{lemma}

\begin{proof}
    This statement follows trivially from the commutation relations
    of $\galg_{n}$~\eqref{eq:gncommrel}.
\end{proof}

\begin{remark}
    The radical $\mathfrak{r}_{n}$ is a two-step nilpotent Lie algebra which
    can be interpreted as a generalised Heisenberg algebra.  We remark that
    in the case $n=3$, it is exactly the Heisenberg algebra $\halg(1)$, or
    in the notation of~\cite{SnobWinternitz2017book} the nilpotent algebra
    $\mathfrak{n}_{1,3}$.
    \label{eq:heisenberg}
\end{remark}

Since our goal is to determine the Casimir elements of the Lie algebra
$\galg_n$, we begin by establishing how many such invariants the algebra
admits:

\begin{proposition}
    The Lie algebra $\galg_{n}$ admits
    \begin{equation}
        \nu_{n} = T_{n-2}+1,
        \label{eq:casnum}
    \end{equation}
    Casimir elements of which $T_{n-2}$ are the central elements $z_{i,j}$ and
    only one is nontrivial up to multiplication by scalars and by  central
    elements.
    \label{prop:casnum}
\end{proposition}

\begin{proof}
    The proof of this result relies on the Betrametti--Blasi formula in 
    \Cref{thm:beltramettiblasi}. Consider the commutative variables in
    $\K$: 
    \begin{equation}
        h, x_{-}, x_{+}, y_{1,-}, y_{1,+},\ldots, 
        y_{n-2,-}, y_{n-2,+},
        z_{1,n-2},\ldots,z_{n-2,n-2},
        \label{eq:vars}
    \end{equation}
    associated to the corresponding generators of $\galg_{n}$, and
    taken in this order. Then, using formula~\eqref{eq:tabledec} we can write
    the matrix $A(n)$ in~\eqref{eq:bbformula} in a recursive way as:
    \begin{equation}
        A(n) =
        \left(
        \begin{array}{cccccc}
            A(n-1) & Y(n) & O
            \\
            -Y^{T}(n) & Z(n) & O
            \\
            O & O & O
        \end{array}
        \right),
        \label{eq:Angen}
    \end{equation}
    where $O$ denotes a zero matrix of the appropriate size, and we defined:
    \begin{equation}
        Y(n) :=
        \begin{pmatrix}
            -y_{n-2,-} & y_{n-2,+}
            \\
            0 & y_{n-2,-}
            \\
y_{n-2,+} & 0
            \\
            0 & -z_{1,n-2}
            \\
            z_{1,n-2} & 0
            \\
            0 & 0
            \\
            0 & -z_{2,n-2}
            \\
            z_{2,n-2} & 0
            \\
            0 & 0
            \\
            0 & 0
            \\
            0 & -z_{3,n-2}
            \\
            z_{3,n-3} & 0
            \\
            0 & 0
            \\
            \vdots & \vdots
            \\
            0 & 0
            \\
            0 & -z_{n-3,n-2}
            \\
            z_{n-3,n-2} & 0
        \end{pmatrix},
        \qquad
        Z(n) :=
        \begin{pmatrix}
            0 & -z_{n-2,n-2}
            \\
            z_{n-2,n-2} & 0
        \end{pmatrix},
        \label{eq:abdecvars}
    \end{equation}
    and the initial datum:
    \begin{equation}
        A(2) := 
        \begin{pmatrix}
            0 & -2 x_{-} & 2 x_{+}
            \\
            2 x_{-} & 0 & -h
            \\
            -2x_{+} & h & 0
        \end{pmatrix}.
        \label{eq:A2def}
    \end{equation}
    
    We aim to prove that $\rank A(n) = 2(n-1)$. To do so, we will extract from
    $A(n)$ a $2(n-1)\times 2(n-1)$ submatrix which evaluated for a particular
    choice of the variables~\eqref{eq:vars} will have non-vanishing
    determinant. From the continuity of the determinant this is enough to prove
    the statement.

    Firstly, let us remove all zero lines and rows in $A(n)$: these correspond
    to the coefficients associated with the central elements $z_{i,j}$. This gives
    us  a matrix with $T_{n}-T_{n-2}=2(n-1)+1$ columns and
    rows of the following shape:
    \begin{equation}
        \tilde{A}(n)
        =
        \begin{pmatrix}
            0 & -2 x_{-} & 2 x_{+} & -y_{1,-} & y_{1,+} & \cdots & -y_{n-2,-} & y_{n-2,+}
            \\
            & 0 & -h & 0 & y_{1,-} & \cdots & 0  & y_{n-2,-}
            \\
            & & 0 & y_{1,+} & 0 & \cdots & y_{n-2,+} & 0
            \\
            & & & 0 & -z_{1,1} & \cdots & 0 & -z_{1,n-2}
            \\
            & & & & \ddots & &  & \vdots
            \\
            &&&&&& & -z_{n-2,n-2}
            \\
            &&&&&& & 0
        \end{pmatrix}.
        \label{eq:Atilda}
    \end{equation}
    where we omit the subdiagonal elements since the matrix is skewsymmetric.
    Formula~\eqref{eq:Atilda} can be proved easily using the recursive formula
    for $A(n)$~\eqref{eq:Angen}. This matrix is a skewsymmetric matrix of odd
    dimension, meaning that its determinant must be zero and its rank is
    \emph{at most} $2(n-1)$.

    Now, consider the submatrix of $\tilde{A}(n)$ obtained by deleting the
    third row and third column and evaluating on the following value of the
    variables~\eqref{eq:vars}:
    \begin{equation}
        y_{i,-}=y_{i,+}=0, 
        \qquad
        z_{i,j} =
        \begin{cases}
            0 & \text{if $i\neq j$},
            \\
            1 & \text{if $i=j$},
        \end{cases}
        \label{eq:varsval}
    \end{equation}
    that is, the $2(n-1)\times 2(n-1)$ matrix:
    \begin{equation}
        \tilde{A}_{0}(n) =
        \begin{pmatrix}
            0 & -2x_{-}
            \\
            2x_{-} & 0
            \\
            &&0&-1
            \\
            &&1&0
            \\
            &&&&\ddots
            \\
            &&&&&0&-1
            \\
            &&&&&1&0
        \end{pmatrix}.
        \label{eq:matrixfin}
    \end{equation}
    The determinant of the matrix $\tilde{A}_{0}(n)$ is clearly non-zero, 
    and as stated above the continuity of the deteminant implies the following
    chain of equalities:
    \begin{equation}
        \rank A(n) = \rank \tilde{A}(n)= \rank \tilde{A}_{0} = 2(n-1).
        \label{eq:rankAA0}
    \end{equation}
    This last finding implies formula~\eqref{eq:casnum} using
    Beltrametti--Blasi's formula in \Cref{thm:beltramettiblasi}. Moreover, since the
    number of central elements is exactly $T_{n-2}$, we conclude
    that these are part of the above number, and then only a nontrivial
    Casimir invariant exists (up to a scalar multiple). This ends the proof of
    the proposition.
\end{proof}

\subsection{Matrix representations of $\galg_{n}$}

In this subsection, following the ideas of~\cite{Zhang_et_al1990} for the
$\mathfrak{h}_{6}$ Lie algebra, we construct a faithful representation of
$\galg_{n}$ in $\Sl_{2n-2}(\K)$. The faithful representation also implies that
the Jacobi identity holds in $\galg_n$. This is the content of the following
proposition:

\begin{proposition}
    The Lie algebra $\galg_{n}$ has the following faithful representation:
    \begin{equation}
        \begin{tikzcd}[row sep=tiny]
            \rho \colon\galg_{n} \arrow{r} & \Sl_{2(n-1)}(\K)
            \\ 
            \phantom{\rho\colon}x \arrow[mapsto]{r} & X.
        \end{tikzcd}
        \label{eq:faithful}
    \end{equation}
    where
    \begin{equation}
         x = ah+bx_++cx_-
        +\sum_{i=1}^{n-2} (d_iy_{i,+}+e_iy_{i,-})
        +\sum_{1\leq i\leq j}^{n-2} m_{ij}z_{i,j},
        \label{eq:elementxi}
    \end{equation}
    and
    \begin{equation}
        X = \left( \begin{array}{ccccccccc}
        0& 0&\cdots& 0& 0&0&0&\cdots&0
        \\
        \vdots& &&&&\vdots& &&\vdots
        \\
        0& 0&\cdots& 0& 0&0&0&\cdots&0
        \\
         d_1&d_2&\cdots&d_{n-2}&a&b&0&\cdots&0
         \\
         e_1&e_2&\cdots&e_{n-2}&c&-a&0&\cdots&0
         \\
         2m_{11}&m_{12}&\cdots&m_{1,n-2}&-e_1&d_1&0&\cdots&0
         \\
         \vdots& &&&&\vdots&&&\vdots
         \\
         m_{1,n-2}& m_{2,n-2}&\cdots& 2m_{n-2,n-2}&-e_{n-2}&d_{n-2}&0&\cdots&0
         \\
        \end{array}\right).
        \label{eq:elementX}
    \end{equation}
    \label{prop:representation}
\end{proposition}

\begin{proof}
    We start by noting that the subspace of $\Sl_{2(n-1)}(\K)$ consisting of
    the matrices $X$ of the form~\eqref{eq:elementX} can be described
    synthetically as:
    \begin{equation}
        X = \left( \begin{array}{ccc}
            0 & 0 & 0
            \\
            P & A & 0 
            \\
            M&Q&0
        \end{array} \right),
        \label{eq:Xsynt}
    \end{equation}
    where $A$ is the $2\times 2$ block representing $\Sl_{2}(\K)$:
    \begin{equation}
    A = 
    \begin{pmatrix}
        a & b
        \\
        c & -a
    \end{pmatrix},
        \label{eq:Ablock}
    \end{equation}
    $P$ is the $(n-2)\times 2$ block:
    \begin{equation}
        P\,=\,\begin{pmatrix} d_1&\ldots&d_{n-2}\\e_1&\ldots&e_{n-2}\end{pmatrix},
        \label{eq:Pblock}
    \end{equation}
    and $Q$ is a $2\times (n-2)$ block related
    to $P$ by the formula:
    \begin{equation}
        Q=P^{T}J,
        \qquad 
        J=\begin{pmatrix}
            0&1
            \\
            -1&0
        \end{pmatrix},
        \label{eq:Qdef}
    \end{equation}
    and finally $M$ is a symmetric matrix.
    We will prove that this subspace is a Lie subalgebra of $\Sl_{2(n-1)}(\K)$, and
    that its commutation relations agree with those of $\galg_{n}$, thus
    proving the proposition. 

    We start by showing that the subspace is closed under the Lie bracket of
    $\Sl_{2(n-1)}(\K)$. By definition, $[X,X']=XX'-X'X$, where the matrix $X'$ has
    submatrices denoted with primes. So, after some easy algebraic
    manipulations we obtain:
    \begin{equation}
        [X,X'] = 
        \left( 
            \begin{array}{ccc}
            0 & 0 &0
            \\
            AP'-A'P & AA'-A'A&0 
            \\
            QP'-Q'P&QA'-Q'A&0
        \end{array} \right).
        \label{eq:XXpcomm}
    \end{equation}
    Observing that if $A,A'\in \Sl_{2}(\K)$ then $AA'-A'A\in\Sl_{2}(\K)$, we
    are left to check that the submatrices in~\eqref{eq:XXpcomm} satisfy the
    two conditions:
    \begin{subequations}
        \begin{align}
            QA'-Q'A &= {\left(AP'-A'P\right)}^{T}J,
            \label{eq:condsa}
            \\
            {\left(QP'-Q'P\right)}^{T} &= QP'-Q'P.
            \label{eq:condsb}
        \end{align}
        \label{eq:conds}
    \end{subequations}
    To prove \eqref{eq:condsb} it is enough to observe that
    $J$ is a skew-symmetric matrix, so that:
    \begin{equation}
    {(QP'-Q'P)}^{T}= \left(P'\right)^{T}(PJ)^{T}- P^{T}\left(P'J\right)^{T} 
    =
    -\left(P'\right)^{T} J P + P^{T}JP' =-Q'P+QP'.
        \label{eq:condsbproof}
    \end{equation}
    On the other hand, to prove~\eqref{eq:condsa}, we observe that $A\in 
    \Sl_{2}(\K)$
    one has $J A = - A^{T} J$, so that:
    \begin{equation}
        QA'-Q'A = P^{T}JA' - {\left(P'\right)}^{T}JA = 
        -P^{T}{\left(A'\right)}^{T}J + {\left(P'\right)}^{T} A^{T}J =
        {\left(AP'-A'P\right)}^{T}J.
        \label{eq:condsaproof}
    \end{equation}
    
    The commutation relations can be proved using the basic matrices $E_{i,j}$
    introduced in equation~\eqref{eq:Eij} and their commutation relations, see
    e.g.~\cite[Section IV.6]{Jacobson1962}:
    \begin{equation}
        [E_{i,j}, E_{k,l}] = \delta_{jk} E_{i,l}-\delta_{i l} E_{k,j}.
        \label{eq:EijEkm}
    \end{equation}
    Indeed, it is easy to see that:
    \begin{equation}
        \begin{gathered}
            \rho(x_+) = E_{n-1,n}, 
            \quad
            \rho(x_-) = E_{n,n-1}, 
            \quad
            \rho(h) = E_{n-1, n-1}- E_{n,n},
            \\
            \rho(y_{i,+}) = E_{n-1,i} + E_{n+i, n},
            \quad
            \rho(y_{i,-}) = E_{n,i} -E_{n+i, n-1},
            \quad
            i=1,\ldots,n-2,
            \\
            \rho(z_{i,j}) = E_{i+n, j} + E_{j+n,i},
            \quad
            1 \leq i \leq j \leq n-2.
        \end{gathered}
        \label{eq:gnEij}
    \end{equation}
    Note that the elements $\rho(z_{i,j})$ automatically commute with all the
    others since the indices of their base matrices never intersect with the
    ones of the others and equation~\eqref{eq:EijEkm} vanishes in such a case.
    As a title of example let us compute the commutation relations of $h$ with
    the other elements:
    \begin{subequations}
        \begin{align}
            [\rho(h),\rho(x_{+})]& 
            \begin{aligned}[t]
                &= [E_{n-1,n-1}-E_{n,n},E_{n-1,n}]
                = [E_{n-1,n-1},E_{n-1,n}]-[E_{n,n},E_{n-1,n}]
                \\
                &=\delta_{n-1,n-1}E_{n-1,n}-\delta_{n-1,n}E_{n-1,n-1}
                -(\delta_{n,n-1}E_{n,n}-\delta_{n,n}E_{n-1,n})
                \\
                &= E_{n-1,n}+E_{n-1,n} = 2 E_{n-1,n}=2\rho(x_{+}),
            \end{aligned}
            \label{eq:HXpcomm}
            \\
            [\rho(h),\rho(x_{-})]& 
            \begin{aligned}[t]
                &= [E_{n-1,n-1}-E_{n,n},E_{n,n-1}]
                \\
                &= - E_{n,n-1}-E_{n,n-1} = -2 E_{n,n-1}=-2\rho(x_{-}),
            \end{aligned}
            \label{eq:HXmcomm}
            \\
            [\rho(h),\rho(y_{i,+})]&
            \begin{aligned}[t]
                &=[E_{n-1,n-1}-E_{n,n},E_{n-1,i}+E_{n+i,n}]
                \\
                &=[E_{n-1,n-1},E_{n-1,i}]-[E_{n,n},E_{n+i,n}]
                \\
                &=\delta_{n-1,n-1}E_{n-1,i}-\delta_{n-1,i}E_{n-1,n-1}
                -(\delta_{n,n+i}E_{n,n}-\delta_{n,n}E_{n+i,n})
                \\
                &=E_{n-1,i}+E_{n+i,n} = \rho(y_{i,+}),
            \end{aligned}
            \label{eq:HYipcomm}
            \\
            [\rho(h),\rho(y_{i,-})]&
            \begin{aligned}[t]
                &=[E_{n-1,n-1}-E_{n,n},E_{n,i}+E_{n+i,n-1}]
                \\
                &=-[E_{n,n},E_{n,i}]-[E_{n-1,n-1},E_{n+i,n-1}]
                \\
                &=-E_{n,i}+E_{n+i,n-1} = -\rho(y_{i,-}).
            \end{aligned}
            \label{eq:HYimcomm}
        \end{align}
        \label{eq:commproof}%
    \end{subequations}
    Then, all the commutation relations of the elements with $h$ coincide
    with those given in~\eqref{eq:gncommrel}. As a last check, we prove the
    commutation relations of the elements $y_{i,+}$ with the elements $y_{j,-}$:
    \begin{equation}
        \begin{aligned}
            [\rho(y_{i,+}),\rho(y_{j,-})] & = [E_{n-1,i}+E_{n+i,n},E_{n,j}-E_{n+j,n-1}]
            \\
            &=[E_{n+i,n},E_{n,j}] - [E_{n-1,i},E_{n+j,n-1}]
            \\
            &=\delta_{n,n}E_{n+i,j}-\delta_{n+i,j}E_{n,n}
                -(\delta_{i,n+j}E_{n-1,n-1}-\delta_{n-1,n-1}E_{n+j,i})
            \\
            &=E_{n+i,j}+E_{n+j,i} = \rho(z_{\min(i,j),\max(i,j)}).
        \end{aligned}
        \label{eq:YipYimcomm}
    \end{equation}
    The last equality in equation~\eqref{eq:YipYimcomm} follows from the
    fact that formula~\eqref{eq:gnEij} is defined for  $1\leq i\leq j \leq n-2$.
    So, also these commutation relations coincide with those given in~\eqref{eq:gncommrel}.
    The remaining commutation relations can be proved in the same way, thus
    ending the proof of the proposition.
\end{proof}

We conclude this subsection with the following corollary which gives a faithful
representation $\overline{\rho}$ of the quotient algebra $\overline{\galg}_{n}$
with image in the $n\times n$ matrices with trace zero. We leave the proof to
the reader, but we observe that the matrix $\overline{\rho}(\overline{x})$ is
the upper left block in the matrix $\rho(x)$.

\begin{corollary}
    The quotient algebra $\overline{\galg}_{n}=\galg_n/Z(\galg_n)$ has the
    following representation in $\mathfrak{sl}_{n}(\mathbb{K})$:
    \begin{equation}
        \begin{tikzcd}[row sep=tiny]
            \overline{\rho} \colon\overline{\galg}_{n} \arrow{r} & \Sl_{n}(\K)
            \\
            \overline{x} \arrow[mapsto]{r} & \overline{X},
        \end{tikzcd}
        \label{eq:faithfulbar}
    \end{equation}
    where $\overline{x}$ has representative
    \begin{equation}
        x\,=\,ah+bx_++cx_-
        +\sum_{i=1}^{n-2} (d_iy_{i,+}+e_iy_{i,-}), 
        \label{eq:elementxibar}
    \end{equation}
    and
    \begin{equation}
        \overline{X} =
    \begin{pmatrix} 0& 0&\ldots& 0& 0&0\\
     \vdots& &&&&\vdots\\
     0& 0&\ldots& 0& 0&0\\
     d_1&d_2&\ldots&d_{n-2}&a&b\\
     e_1&e_2&\ldots&e_{n-2}&c&-a
    \end{pmatrix}~.
        \label{eq:slnrepr}
    \end{equation}
    \label{cor:gbarrepr}
\end{corollary}

\subsection{The coadjoint representation in the space of vector
fields}\label{advec}

We conclude this section on the generalities of the Lie algebra  $\galg_{n}$ by
presenting the explicit form of its coadjoint representation in the space of
vector fields on $\galg_n^*$. This is the content of the following proposition:

\begin{proposition}
    The coadjoint representation of the Lie algebra $\galg_{n}$
    in the space of vector fields $\mathfrak{X}(\galg^*)$
    is generated by the following vector fields:
    \begin{subequations}
        \begin{align}
            \widehat{h} &= 
            -2x_{-}\frac{\partial}{\partial x_{-}}
            +2x_{+}\frac{\partial}{\partial x_{+}}
            +\sum_{j=1}^{n-2}
            \left[y_{j,+}\frac{\partial}{\partial y_{j,+}}
                -y_{j,-}\frac{\partial}{\partial y_{j,-}}\right],
            \label{eq:coadjh}
            \\
            \widehat{x}_{\mp} &= 
            \pm 2x_{\mp}\frac{\partial}{\partial h}
            \mp h \frac{\partial}{\partial x_{\pm}}
            +\sum_{j=1}^{n-2}y_{j,\mp}\frac{\partial}{\partial y_{j,\pm}},
            \label{eq:coadjxp}
            \\
            \widehat{y}_{i,\mp} &= 
            \pm y_{i,\mp}\frac{\partial}{\partial h}
            - y_{i,\pm} \frac{\partial}{\partial x_{\pm}}
            \mp \sum_{j=1}^{n-2}z_{j,i}\frac{\partial}{\partial y_{j,\pm}},
            \label{eq:coadjypm}
            \\
            \widehat{z}_{i,j} &= 0.
            \label{eq:coadjzij}
        \end{align}
        \label{eq:coadjgn}
    \end{subequations}
    \label{prop:coadjgn}
\end{proposition}

The proof of \Cref{prop:coadjgn} is an immediate consequence of the general  result presented in equation~\eqref{eq:eihat}.  For example, from the commutation
relations \eqref{eq:gncommrel}, one has $\rho_{\adj}(x_-)(x)=[x_-,x]\neq 0$
for basis element $x$ of $\galg_n$ only if $x=h,x_+,y_{i,+}$ and one has
\begin{equation}
    [x_-,h]=2x_-,\quad [x_-,x_+]\,=\,-h,\quad [x_-,y_{i,+}]=y_{i,-}~.
    \label{eq:hcomm}
\end{equation}
Thus the vector field $\widehat{x}_-$ is given by
\begin{equation}
    \widehat{x}_{-} \,=\,
    2x_{-}\frac{\partial}{\partial h}
    - h \frac{\partial}{\partial x_{+}}
    +\sum_{j=1}^{n-2}y_{j,-}\frac{\partial}{\partial y_{j,+}},
    \label{eq:vfieledxm}
\end{equation}
since $\widehat{x}_-(x)=[x_-,x]$ for any $x\in \galg_n$.

We observe that, as expected, the coadjoint representation is not faithful
because the central elements all map to the null vector field.  Thus, the
coadjoint representation is in fact a representation of
$\overline{\galg}_{n}=\galg_n/Z(\galg_n)$.

\section{Construction of the Casimir invariants of $\galg_{n}$}
\label{sec:casgn}

\noindent In this section we review the construction of the Casimir elements
for the Lie algebras $\galg_2 \cong\, \mathfrak{sl}_{2}(\mathbb{K})$ and
$\galg_3 \cong \,\halg_{6}$ with the standard method of vector fields, and we
computationally extend the result to the case of $\galg_{4}$. From these
low-dimensional cases, we are able to guess the general form of the Casimir
polynomial. After a short digression on the properties of the $\Sl_{n}(\K)$ Lie
algebra and its invariants we prove the guess in the main theorem of this
section and of the whole paper.

\subsection{Low dimensional cases: $n=2,3,4$}
\label{sss:lowdim}

Let us start with the Lie algebra $\galg_{2}\cong \Sl_{2}(\mathbb{K})$, and let
us find its Casimir element using the coadjoint action on the space of vector
fields. In such a case we have to solve the following system of overdetermined
PDEs for the function $F=F(h,x_{-},x_{+})$, obtained from \eqref{eq:coadjgn}
with $n=2$:
\begin{equation}
    \widehat{h}(F) = 
    -2x_{-}\frac{\partial F}{\partial x_{-}}
    +2x_{+}\frac{\partial F}{\partial x_{+}}=0,
    \quad
    \widehat{x}_{\mp}(F) = 
    \pm 2x_{\mp}\frac{\partial F}{\partial h}
    \mp h \frac{\partial F}{\partial x_{\pm}} =0.
    \label{eq:C2find}
\end{equation}
Solving these three equations using for instance the computer algebra system
\texttt{Maple} we get the following solution:
\begin{equation}
    F = F\left( h^{2} + 4 x_{+}x_{-}  \right),
    \label{eq:C2sol}
\end{equation}
where $F$ is an arbitrary function of its argument,  to 
which corresponds the only functionally independent polynomial 
Casimir element in $S\galg_{2}$:
\begin{equation}
C_{2} = h^{2} + 4 x_{+}x_{-}.
\label{eq:C2solh}
\end{equation}
This is of course a well-know result, see for 
instance~\cite[\S16.5]{SnobWinternitz2017book}.

We observe that one can write the (functional independent) Casimir of
$S\galg_{2}$~\eqref{eq:C2solh} as the determinant of a $2\times2$ symmetric
matrix in the following way:
\begin{equation}
    C_2 = - \det 
    \begin{pmatrix}
        -2x_{-} & h
        \\
        h & 2x_{+}
    \end{pmatrix}.
    \label{eq:C2mat}
\end{equation}

\begin{remark}
    As stated in \Cref{sec:background} Casimir elements in $S\galg$ are usually
    found just by looking at homogeneous polynomial solutions of the system of
    PDEs~\eqref{eq:pdej}. In this simple case, we choose to present the general
    solution of the system~\eqref{eq:C2find} since it is one of the few cases
    where using a computer algebra system it is possible to obtain such a
    general solution in a reasonable amount of time.
    \label{rem:polcas}
\end{remark}

Let us now consider the next Lie algebra in the hierarchy, namely
$\galg_{3}\cong \halg_{6}$. 
To find the Casimir elements in $S\galg_{3}$,
we use again the PDE system obtained from the coadjoint
representation from \eqref{eq:coadjgn} with $n=3$:
\begin{subequations}
    \begin{align}
        \widehat{h}(F) &= 
        -2x_{-}\frac{\partial  F}{\partial x_{-}}
        +2x_{+}\frac{\partial F}{\partial x_{+}}
        +y_{1,+}\frac{\partial F}{\partial y_{1,+}}
        -y_{1,-}\frac{\partial F}{\partial y_{1,-}}
        =0,
        \\
        \widehat{x}_{\mp}(F) &= 
        \pm 2x_{\mp}\frac{\partial F}{\partial h}
        \mp h \frac{\partial F}{\partial x_{\pm}} 
        + y_{1,\mp} \frac{\partial F}{\partial y_{1,\pm}} 
        =0,
        \\
        \widehat{y}_{1,\pm}(F) &= 
        \pm y_{1,\mp}\frac{\partial F}{\partial h}
        - y_{1,\pm} \frac{\partial F}{\partial x_{\pm}} 
        \mp z_{1,1} \frac{\partial F}{\partial y_{1,\pm}} 
        =0,
        \\
        \widehat{z}_{1,1}(F) & = 0,
    \end{align}
    \label{eq:C3find}
\end{subequations}
for the unknown function $F=F(h,x_{-},x_{+},y_{1,-},y_{1,+},z_{1,1})$. Now,
instead of a generic function $F$, since we work in the symmetric algebra
$S\galg_n$ we can restrict our search to homogeneous Casimir polynomials $P$ of
degree $l$, namely:
\begin{equation}
    P^{(l)} = \sum_{i_{1}+\ldots+i_{6}=l}
    c_{i_{1},\ldots,i_{6}}
    h^{i_{1}}x_{-}^{i_{2}}x_{+}^{i_{3}}
    y_{1,-}^{i_{4}}y_{1,+}^{i_{5}}z_{1,1}^{i_{6}}.
    \label{eq:C3d}
\end{equation}
Substituting the expression~\eqref{eq:C3d} into the
system~\eqref{eq:C3find} for $l=1,2$ we obtain only
trivial solutions involving the central element $z_{1,1}$, but we
find, up to an unessential multiplicative constant, the following degree
three Casimir element:
\begin{equation}
    C_{3}= 2 x_{+}y_{1,-}^{2}
    -2x_{-}y_{1,+}^2 +2 h y_{1,-}y_{1,+}
    + z_{1,1} (h^2  + 4 x_{-} x_{+}),
    \label{eq:C3sol}
\end{equation}
where $C_3=P^{(3)}|_{\text{solution of} \,\,\eqref{eq:C3find}}$. 
This Casimir element $C_{3}$ corresponds, mutatis mutandis, to the known
Casimir element of the two-photon Lie algebra $\halg_{6}$, see for
instance~\cite{BallesterosHerranz2001,BlascoPhD,BallesterosBlasco2008}.

Even more interestingly, we observe that that it is possible to write this
degree three Casimir of $S\galg_{3}$~\eqref{eq:C3sol} as the determinant of a
$3\times 3$ symmetric matrix with coefficients in $\galg_3\subset S \galg_3 =
\Pol(\galg_3^*)$:
\begin{equation}
   C_{3} = 
    -\det
    \begin{pmatrix}
        z_{1,1} & -y_{1,-} & y_{1,+}
        \\
        -y_{1,-} & -2x_{-} & h
        \\
        y_{1,+} & h & 2x_{+}
    \end{pmatrix}.
    \label{eq:C3soldet}
\end{equation}
The idea of considering a determinant came from a study of this polynomial
Casimir. As a homogeneous cubic in six variables it defines a hypersurface in
$\PP^5$. One finds that it has a 2-dimensional singular locus, which is rather
special. At this point one remembers that the determinant of a $3\times3$
symmetric matrix defines a similar hypersurface, which is singular in the
symmetric matrices of rank one, this gives a $2$-dimensional singular locus.
With some trial and error one find the expression for $C_3$ given above.

In the case $n=4$, for $l=1,2,3,4$, we have to
look for homogeneous polynomial solutions of the system of PDEs
\begin{equation}
    \begin{aligned}
        \widehat{h}(P^{(l)}) &= 
        -2x_{-}\frac{\partial P^{(l)}}{\partial x_{-}}
        +2x_{+}\frac{\partial P^{(l)}}{\partial x_{+}}
        +\sum_{j=1}^{2}\left[y_{j,+}\frac{\partial P^{(l)}}{\partial y_{j,+}}
        -y_{j,-}\frac{\partial P^{(l)}}{\partial y_{j,-}}\right]
        =0,
        \\
        \widehat{x}_{\mp}(P^{(l)}) &= 
        \pm 2x_{\mp}\frac{\partial P^{(l)}}{\partial h}
        \mp h \frac{\partial P^{(l)}}{\partial x_{\pm}} 
        + \sum_{j=1}^{2}y_{j,\mp} \frac{\partial P^{(l)}}{\partial y_{j,\pm}} 
        =0,
        \\
        \widehat{y}_{i,\pm}(P^{(l)})) &= 
        \pm y_{i,\mp}\frac{\partial P^{(l)}}{\partial h}
        - y_{i,\pm} \frac{\partial P^{(l)}}{\partial x_{\pm}}
        \mp \sum_{j=1}^{2} z_{j,i} \frac{\partial P^{(l)}}{\partial
        y_{j,\pm}}=0, \quad i=1,2 
        \\
        \widehat{z}_{1,1}(P^{(l)}) & = 
        \widehat{z}_{1,2}(P^{(l)}) = \widehat{z}_{2,2}(P^{(l)}) = 0 .
    \end{aligned}
    \label{eq:C4find}
\end{equation}
With the same notation as above, also in this case for $l=1,2,3$ one
finds trivial solutions related to the centre of $\galg_{4}$, and finally a
new non-trivial element for $l=4$. This non-trivial element can be
written down as the determinant of a symmetric matrix with coefficients
in $\galg_4\subset S\galg_4=\Pol(\galg_4^*)$:
\begin{equation}
    C_{4}=
    -\det
    \begin{pmatrix}
        z_{1,1} & z_{1,2} & -y_{1,-} & y_{1,+}
        \\
        z_{1,2} & z_{2,2} & -y_{2,-} & y_{2,+}
        \\
        -y_{1,-} & -y_{2,-} & -2x_{-} & h
        \\
        y_{1,+} & y_{2,+} & h & 2x_{+}
    \end{pmatrix}.
    \label{eq:C4soldet}
\end{equation}
So, from the result above one can guess that the form of the Casimir
polynomial for general $n \geq 2$ is:
\begin{equation}
    C_{n} = -\det(M_n),
    \label{eq:casn}
\end{equation}
where $M_{n}$ is the following $n \times n$ symmetric matrix, with coefficients
in $\galg_n$:
\begin{equation}
    M_n\,:=\,\begin{pmatrix}
        z_{1,1} & \cdots  & z_{1,n-2} & -y_{1,-} & y_{1,+}
        \\
        \vdots & \ddots & \vdots & \vdots & \vdots
        \\
        z_{1,n-2} & \cdots & z_{n-2,n-2} & -y_{n-2,-} & y_{n-2,+}
        \\
        -y_{1,-} & \cdots & -y_{n-2,-} & -2x_{-} & h
        \\
        y_{1,+} & \cdots & y_{n-2,+} & h & 2x_{+}
    \end{pmatrix}~.
    \label{eq:Mn0}
\end{equation}
Experimentally, one can check that this formula provides the Casimir
polynomials also for $n=5,6,7,8$. However, by simple observation it is not
possible to establish that this will hold for all $n$. To do so requires a
change of perspective.

\subsection{The $\Sl_n(\K)$ case}

In this subsection we recall the crucial facts about determinants
and $\SL_{n}(\K)$ invariance, interpreted in terms of representation
theory of $\Sl_{n}(\K)$.

The determinant of a symmetric $n\times n$ matrix $N$ is a homogeneous
polynomial of degree $n$ in the coefficients $(N)_{ij}=n_{ij}$, $1\leq i\leq j\leq n$,
of $N$.  It is invariant under the map $N\mapsto ANA^T$ for $A\in \SL_{n}(\K)$,
that is $\det(N)=\det(ANA^T)$. The vector space of symmetric matrices can be
identified with $\Sym^2(\K^n)$ and $N\mapsto ANA^T$ is the representation of the
Lie group $\SL_n(\K)$ (for $\K=\R$ or $\CC$) on $\Sym^2(\K^n)$.

Let $A(t)\in \SL_n(\K)$ be a one parameter subgroup with $A(0)=I$, $A'(0)=X\in
\Sl_n(\K)$ and $A(s+t)=A(s)A(t)$. Differentiating the matrix $A(t)N{A(t)}^T$ with
respect to $t$ and evaluating at $t=0$ gives:
\begin{equation}
    \eval{\dv{t} A(t)NA(t)^T}_{t=0}\,=\,XN\,+\,NX^T~,
\end{equation}
which is indeed $\rho^{(2)}(X)(N)$ for the representation $\rho^{(2)}$ of
$\Sl_n(\K)$ on $\Sym^2(\K^n)$, see~\eqref{eq:rhot2}.

We now consider the coefficients $n_{ij}$, $1\leq i\leq j\leq n$, of a
symmetric $n\times n$ matrix $N$ as independent polynomial variables.  Then, the
representation $\rho^{(2)}$ on the linear forms $n_{ij}$ above induces, by
derivations, a representation of $\Sl_n(\K)$ on the polynomials in
$\K[\ldots,n_{ij},\ldots]$. Thus, as we have seen in \Cref{sec:background}, we get a
representation of $\Sl_n(\K)$ in the vector fields on $\K^{n(n+1)/2}$, defined as:
\begin{equation}
    \begin{tikzcd}[row sep=tiny]
        \Sl_n(\K)  \arrow{r} & \mathfrak{X}(\K^{n(n+1)/2})
        \\ 
        X \arrow[mapsto]{r} & \widehat{X}\,:=\displaystyle\,\sum_{1\leq i\leq j\leq n}
        \ell_{ij}(X)\frac{\partial}{\partial n_{ij}}~,
    \end{tikzcd}
\end{equation}
where:
\begin{equation}
    \ell_{ij}(X)\,:=\,(XN\,+\,NX^T)_{ij}~.
    \label{eq:ellij}
\end{equation}
This representation has the property that for $G\in \K[\ldots,n_{ij},\ldots]$ one has
\begin{equation}
    \widehat{X}(G)\,=\, \eval{\dv{t} G(A(t)N{A(t)}^T)}_{t=0} =\,
    \sum_{1\leq i\leq j\leq n} \frac{\partial G}{\partial n_{ij}} (XN+NX^T)_{ij}~,
    \label{eq:XGact}
\end{equation}
where $A(t)=\exp(tX)\in \SL_n(\K)$ is the one parameter subgroup generated by
$X\in \Sl_n(\K)$.

The function $\det(N)\in \K[\ldots,n_{ij},\ldots]$ is a polynomial, of degree
$n$, in the $n_{ij}$. The function $\det(N)$, being invariant under the action
of $\SL_n(\K)$ on $N$ so $\det(A(t)NA(t)^T)$ is constant in $t$, is annihilated by the vector fields $\widehat{X}$:
\begin{equation}
    \widehat{X}(\det(N))\,=\,0,\qquad \forall\,X\,\in\,\Sl_n(\K)~.
\end{equation}

\begin{example}
    We consider the case $n=2$. The Lie algebra $\Sl_n(\K)$ consists of the
    matrices with trace zero, so we have:
    \begin{equation}
        N\,=\,
        \begin{pmatrix} 
            n_{11}&n_{12}
            \\
            n_{12}&n_{22}
        \end{pmatrix},
        \qquad
        \Sl_2(\K)\,=\,\left\{X=\begin{pmatrix} a&b\\c&-a\end{pmatrix}:\;a,b,c\,\in\,\K\,\right\}~.
        \label{eq:sl2case}
    \end{equation}
    Since:
    \begin{equation}
        XN\,+\,NX^T\,=\,
        \begin{pmatrix}
            2an_{11}+2bn_{12}&cn_{11}+bn_{22}\\
            cn_{11}+bn_{22}&2cn_{12}-2an_{22}
        \end{pmatrix}
        \label{eq:XNNXT2}
    \end{equation}
    we find for $X\in \Sl_2(\K)$ the vector field in $\mathfrak{X}(\K^{3})$:
    \begin{equation}
        \widehat{X}=\,2(a n_{11}+b n_{12})\frac{\partial}{\partial  n_{11}}\,+\,
    (c n_{11}+b n_{22})\frac{\partial}{\partial  n_{12}}\,+\,2(c n_{12}-a n_{22})
    \frac{\partial}{\partial  n_{22}},
    \end{equation}
    or more explicitly considering
    the standard basis of $\Sl_{2}(\K)$ we have three vector fields:
    \begin{subequations}
        \begin{align}
            \widehat{X}_{h} &= 2 n_{11}\frac{\partial}{\partial  n_{11}}\,-\,2 n_{22}\frac{\partial}{\partial  n_{22}}
            \label{eq:Xhata}
            \\
            \widehat{X}_{x_+} &=2n_{12}\frac{\partial}{\partial  n_{11}}\,+\,
            n_{22}\frac{\partial}{\partial  n_{12}},
            \label{eq:Xhatb}
            \\
            \widehat{X}_{x_-} &= n_{11}\frac{\partial}{\partial  n_{12}}\,+\,
            2n_{12}\frac{\partial}{\partial  n_{22}},
            \label{eq:Xhatc}
        \end{align}
        \label{eq:Xhatabc}
    \end{subequations}
    such that $\widehat{X}=a\widehat{X}_{h}+b\widehat{X}_{x_+}+c\widehat{X}_{x_-}$.
    As guaranteed by the way we defined these vector fields, we then have
    \begin{equation}
        \widehat{X}_{x}(\det N)\,=\, \widehat{X}_{x}( n_{11} n_{22}- n_{12}^2)\,=\,0,
        \quad x=h,x_+,x_-,
        \label{eq:XdeN2}
    \end{equation}
    and this is of course also easy to verify. This is clearly equivalent,
    up to a change of basis, to the vector field approach we presented
    in \Cref{sec:casgn}.
\end{example}

\subsection{The $\galg_n$ case}

Now we prove the following statement:

\begin{theorem}
    The functionally independent Casimir invariants of the Lie algebra
    $\mathfrak{g}_{n}$ are the linear central elements $z_{i,j}$ and the
    $n$th degree polynomial $C_{n}$ given in \Cref{eq:casn}.
    \label{thm:cas}
\end{theorem}

To prove \Cref{thm:cas} we need to show that the determinant of the matrix
$M_{n}$~\eqref{eq:Mn0} is annihilated by the vector field $\widehat{x}$ on
$\galg_n^*$ for all $x\in \galg_n$.  Since $M_{n}$ is a symmetric
matrix, we want to apply the ideas discussed in the previous Subsection. 

Notice that $\galg_{n}\subset S\galg_n=\K[h,x_-,x_{+},\ldots,z_{n-2,n-2}]$ and
that each of the $n(n+1)/2$ variables appears as a coefficient  (up to
constants) of $M_n$, so that $M_n$ is basically the matrix $N$ above, if we
make the appropriate identification of the variables.

We will show in Lemma~\ref{Xx} that for each $x\in\galg_{n}$ there is a matrix
$X_x\in \Sl_n(\K)$ such that $\rho_{\adj}(x)(M_n)=X_xM_n+M_nX_x$.  Here by
$\rho_{\adj}(x)(M_n)$ we mean the matrix obtained by applying
$\rho_{\adj}(x)$ to all coefficients of $M_n$. From this we will easily deduce
that $\hat{x}(\det(M_n))=0$.

Since the $z_{i,j}$ are in the center of $\galg_n$, which is the kernel of
adjoint representation, the vector fields $\widehat{z}_{i,j}$ are zero. Thus we
only need to consider the quotient Lie algebra $\overline{\galg}_n$.
This quotient algebra has a faithful representation in $\Sl_n(\K)$, see
\Cref{cor:gbarrepr}. We will find it convenient to consider the
composition:
\begin{equation}
    \begin{tikzcd}
    \galg_{n}\arrow[r] \arrow[dr] & \overline{\galg}_{n} 
        \arrow[d]
        \\
        & \Sl_{n}(\K)
    \end{tikzcd}
    \label{eq:comp}
\end{equation}
again denoted by $\overline{\rho}$, and whose action is:
\begin{equation}
    \begin{tikzcd}[row sep=tiny]
        \overline{\rho} \colon\galg_{n} \arrow{r} & \Sl_{n}(\K)
        \\
        \phantom{\overline{\rho} \colon}x \arrow[mapsto]{r} & \overline{\rho}(x)\,=\,\begin{pmatrix}0&0\\P&A\end{pmatrix}~,
    \end{tikzcd}
    \label{eq:rhobar}
\end{equation}
where:
\begin{equation}
    x\,=\,ah+bx_++cx_-+\sum_{j=1}^{n-2}(d_iy_{i,+}+e_iy_{i,-}),
\end{equation}
with $P$ and $A$ defined by the equations~\eqref{eq:Ablock}
and~\eqref{eq:Pblock}, respectively. We now show that
$X_x=-\overline{\rho}(x)$.

\begin{lemma}\label{Xx}
    For $x\in\galg_{n}$ we have:
    \begin{equation}
        \rho_{\adj}(x)(M_n)\,=\,-\big(\overline{\rho}(x)M_n\,+\,M_n\overline{\rho}(x)^T\big)~.
        \label{eq:rhoadjxMn}
    \end{equation}
\end{lemma}

\begin{proof}
    For convenience, we write $M_n$ in block form
    \begin{equation}
       M_{n} =
        \begin{pmatrix}
            S & U
            \\
            U^{T} & R
        \end{pmatrix},
        \label{eq:Mn}
    \end{equation}
    where the (sub)matrices $S$, $U$, and $R$ are defined as follows:
    \begin{equation}
        S = \begin{pmatrix}
            z_{1,1} & \cdots  & z_{1,n-2}
            \\
            \vdots & \ddots & \vdots
            \\
            z_{1,n-2} & \cdots & z_{n-2,n-2}
        \end{pmatrix},
        \quad
        U =\begin{pmatrix}
            -y_{1,-} & y_{1,+}
            \\
            \vdots & \vdots
            \\
            -y_{n-2,-} & y_{n-2,+}
        \end{pmatrix},
        \quad
        R =\begin{pmatrix}
            -2x_{-} & h
            \\
             h & 2x_{+}
         \end{pmatrix}.
        \label{eq:SURmat}
    \end{equation}
    Note that the matrix $S$ is symmetric. Since
    $\rho_{\adj}(x)(z_{i,j})=[x,z_{i,j}]=0$, the upper left block of
    $\rho_{\adj}(X)(M_n)$ is zero.  For the upper right block of
    $\rho_{\adj}(X)(M_n)$ we use that
    \begin{subequations}
    \begin{align}
    \rho_{\adj}(ah+bx_++cx_-)(-y_{i,-})\,&=\,ay_{i,-}-by_{i,+},\\
    \rho_{\adj}(ah+bx_++cx_-)(y_{i,+})\,&=\,ay_{i,+}+cy_{i,-}~,
    \end{align}
    \end{subequations}
    leading to the term $-UA^T$ in the upper right block and its transpose in
    the lower right block. Similarly:
    \begin{subequations}
    \begin{align}
        \rho_{\adj}\left(\sum_{j=1}^{n-2} (d_j y_{+,j}+e_j y_{j,-})\right)(-y_{i,-})\,
        &=\,-\sum_{j=1}^{n-2} d_j z_{i,j},\\
        \rho_{\adj}\left(\sum_{j=1}^{n-2} (d_j y_{+,j}+e_j y_{j,-})\right)(y_{i,+})\,
        &=\, -\sum_{j=1}^{n-2} e_j z_{i,j},
    \end{align}
    \end{subequations}
    hence $\rho_{\adj}(\sum_{j} (d_j y_{+,i}+e_j y_{j,-}))$ maps $U$ to $-SP^T$. 
    For the lower right block we compute:
    \begin{equation}
        \begin{aligned}
            \rho_{\adj}(ah+bx_++cx_-)(R)\,&=\,
            \begin{pmatrix}
                4ax_--2bh&-2bx_++2cx_-\\-2bx_++2cx_-&2ch-4ax_+
            \end{pmatrix}
            \\
            \,&=\,-\big(AR\,+\,RA^T\big)~.
        \end{aligned}
        \label{eq:rhoadjR1}
    \end{equation}
    Similarly:
    \begin{equation}
        \begin{aligned}
\rho_{\adj}(\sum (d_jy_{+,j}+e_jy_{j,-})(R)\,&=\,
\begin{pmatrix}
2\sum_jd_jy_{j,-}&\sum_j (-d_jy_{j,+}+e_jy_{j,-})\\
\sum_j (-d_jy_{j,+}+e_jy_{j,-})&-2\sum_jd_jy_{j,+}
\end{pmatrix}
        \\
        \,&=\,-\big(PU+U^TP^T\big)~.
        \end{aligned}
        \label{eq:rhoadjR2}
    \end{equation}
    Putting everything together, we showed that:
    \begin{equation}
    \rho_{\adj}(x)(M_n)\,     =\,
        \begin{pmatrix}
            0 & -(U A^{T} + S P^{T})
            \\
            -(A U^{T} + P S) & -(A R + R A^{T} + PU + U^{T}P^{T})
        \end{pmatrix}.
        \label{eq:Mndir}
    \end{equation}
    On the other hand, we can compute the action $\overline{\rho}(x)$ directly
    and we obtain the desired equality:
    \begin{equation}
        \begin{aligned}
            \overline{\rho}(x)M_n  + M_{n}\overline{\rho}(x)^T
            &=
            \begin{pmatrix}
                0 & 0
                \\
                P & A
            \end{pmatrix}
            \begin{pmatrix}
                S & U
                \\
                U^{T} & R
            \end{pmatrix}
            +
            \begin{pmatrix}
                S & U
                \\
                U^{T} & R
            \end{pmatrix}
            \begin{pmatrix}
                0 & P^{T}
                \\
                0 & A^{T}
            \end{pmatrix}
            \\
            &=
            \begin{pmatrix}
                0 & 0
                \\
                PS + AU^{T} & PU +AR
            \end{pmatrix}
            +
            \begin{pmatrix}
                0 & SP^{T} + UA^{T}
                \\
                0 & U^{T}P^{T} + RA^{T}
            \end{pmatrix}
            \\
            &=
            \begin{pmatrix}
                0 & SP^{T} + UA^{T}
                \\
                PS + AU^{T} & PU  + U^{T}P^{T}+AR + RA^{T}
            \end{pmatrix}.\\
            &=\,  - \rho_{\adj}(x)(M_n)~.
        \end{aligned}
        \label{eq:actMn}
    \end{equation}
    The ends the proof of the lemma.
\end{proof}

We are in place to complete the proof of the \Cref{thm:cas}.

\begin{proof}[Proof of \Cref{thm:cas}]
    For $x\in\galg_n$, we define $X_x:=-\overline{\rho}(x)\in \Sl_n(\K)$.  Then
    the vector field $\widehat{X}_x$ has the property that
    $\widehat{X}_x(y)=\rho_{\adj}(x)(y)$ for all $y\in \galg_{n}\subset
    S\galg_{n}$ and thus $\widehat{X}_x=\rho_{\adj}(x)$ for all $x\in \galg_n$.
    For $x\in \galg_n$, let $A(t)=\exp(tX_x) \in\SL_n(\K)$, then:
    \begin{equation}
        \rho_{\adj}(x)(\det M_n)\,=\,\eval{\dv{t} \det(A(t)N{A(t)}^T)}_{t=0}\,=\,0,
        \label{eq:rhoadjf}
    \end{equation}
    where the first equality follows from Lemma~\ref{Xx}, and the second from
    the invariance of the determinant. Thus, $C_n:=-\det(M_n)$ is a polynomial
    Casimir for $\galg_n$.
\end{proof}

In the next section, as a physical application in classical mechanics within
the framework of coalgebra symmetry, we derive a one degree of freedom
canonical symplectic realisation of the Lie-Poisson coalgebra
$(\mathcal{F}(\galg_{n}^*), \{,\}, \Delta)$  and we use the information on its
polynomial invariant to construct Hamiltonian systems in arbitrary dimensions
with a given number of first integrals obtained by the coalgebra symmetry
through the Casimir of degree $n$ reported in \Cref{eq:casn}.

\section{Canonical symplectic realisations and Hamiltonian systems associated to the Lie--Poisson coalgebra constructed from $\mathfrak{g}_n$}
\label{sec:coalg}

\noindent In this section we present a hierarchy of Hamiltonian models arising
from the Lie--Poisson algebra $S\mathfrak{g}_n=\Pol(\galg_n^*) \subset
\mathcal{F}(\galg_n^*)$ using the coalgebra symmetry approach as detailed in
\Cref{sec:background}. To do so, we will consider as a base case the simplest
symplectic manifold possible, namely $\R^{2}$ with local coordinates $(q,p)$,
the canonical symplectic form $\omega=\dd q \wedge \dd p$, and Poisson bracket:
\begin{equation}
    \{f,g\}=\pdv{f}{q}\pdv{g}{p}-\pdv{g}{q}\pdv{f}{p} \,,
    \quad
f,g\in\mathcal{C}^{\infty}(\R^{2}).
\label{eq:canpoiss}
\end{equation}
Our final result will be that to the Lie--Poisson algebra constructed from the
Lie algebra $\galg_{n}$ corresponds a hierarchy of Hamiltonian systems in $N$
degrees of freedom (obtained from a canonical one degree of freedom
realisation) possessing a grand total of $2N-(2n-2)$ first integrals, of which
$N-(n-2)$ are in involution. These Hamiltonian systems will be, in general,
non-integrable. The first two cases, namely $n=2,3$ are well known in the
literature, having being studied in many different
papers~\cite{Ballesteros_et_al_1996,BallesterosRagnisco1998,BallesterosHerranz2001,Ballesteros_et_al2008PhysAtomNuclei,BallesterosBlasco2010}.
Now, we will details this construction, starting from the basic cases $n=2,3,4$
and confront our results with the literature.

\subsection{The case $n=2$: $\mathfrak{g}_2$ coalgebra symmetry}

The case $n=2$ corresponds to the Lie-Poisson coalgebra associated to
$\mathfrak{sl}_2(\mathbb{R})$. The coalgebra symmetry approach is well known in
the literature, see
e.g.~\cite{Ballesteros_et_al2009,Ballesteros_et_al2008PhysAtomNuclei}.  We use
this example to illustrate in details the construction
outlined in \Cref{sss:coalg}. The reader who already knows the method can
safely jump to \Cref{sss:genericn}.

It is possible to show that the following map:
\begin{equation} 
    \begin{tikzcd}[row sep=tiny]
         \DD_2\colon S\galg_{2}  \arrow{r} & 
        \mathcal{C}^{\infty}(\R^{2})
        \\ 
        \phantom{\DD_{2}\colon S\galg}(h,x_{-},x_{+}) \arrow[mapsto]{r} &
        (q p,-q^2/2, p^2/2),
    \end{tikzcd}
    \label{eq:sl2cc}
\end{equation}
give rise to a one degree of freedom realisation of the linear 
elements of the symmetric algebra
$S\galg_{2}$. Indeed, by direct computation we have that:
\begin{equation}
    \{\DD_2(x_+),\DD_2(x_-)\}=\DD_2(h),
    \quad  
    \{\DD_2(h), \DD_2(x_\pm)\}=\pm 2  \DD_2(x_\pm),
\end{equation}
where we used the canonical Poisson bracket~\eqref{eq:canpoiss}. When we consider the canonical realization the Casimir assumes the
value:
\begin{equation}
    C_2(q,p)= \DD_2(C_2)=\DD_2(h)^2+4 \DD_2(x_-)\DD_2( x_+)\equiv0.
\end{equation}
At this level, we can take as one degree of freedom Hamiltonian any polynomial
(and then extend to any smooth/analytic function) of the generators of
the Lie--Poisson algebra $S\galg_2$, i.e.:
\begin{equation}
    H_{2}=H_{2}(\DD_2(h), \DD_2(x_-), \DD_2(x_+)) = H_{2}(qp,-q^{2}/2,p^{2}/2).
	\label{eq:ham1}
\end{equation}
To give a concrete example, the harmonic oscillator arises with the choice:
\begin{equation}
    H_\text{ho}=\DD_2(x_+)-\omega^{2}\DD_2(x_-)
    \label{eq:ham1ho}
\end{equation}
in~\eqref{eq:ham1}. 
Now, we can use the primitive coproduct $\Delta$ to obtain the generators
of {$S\galg_{2}$ within $(S\galg_{2})^{\otimes 2}$, namely:
\begin{equation}
    \Delta(h)=h \otimes 1 +1 \otimes h, 
    \quad   \Delta (x_{\pm})=x_{\pm} \otimes 1 + 1 \otimes x_\pm.
    \label{eq:gensl22}
\end{equation}
Applying this to the Hamiltonian~\eqref{eq:ham1} we obtain its two
degrees of freedom version:
\begin{equation}
    \begin{aligned}
    H_{2}^{(2)}&:=  \Delta(H_{2})=H_{2}(  \Delta(h),   \Delta(x_-),   \Delta(x_+))
        \\
        &=H_{2}(h \otimes 1 +1 \otimes h, x_{-} \otimes 1 + 1 \otimes x_-, x_{+} \otimes 1 + 1 \otimes x_+).
    \end{aligned}
    \label{eq:ham2}
\end{equation}
In the same way we obtain the Casimir element:
\begin{equation}
    \begin{aligned}
        C^{(2)}_2&:=
        \Delta(C_2)=\Delta(h)^2+4 \Delta(x_-)\Delta(x_+)
        \\
        &=C_2 \otimes 1 + 1 \otimes C_2
        +2 \bigl(h \otimes h+2(x_- \otimes x_++x_+ \otimes x_-) \bigl).
    \end{aligned}
\end{equation}
From this last formula we see how the coproduct produces non-trivial elements
by ``intertwining'' elements coming from different sites. 

Now, we come to the realisation. By direct computation we see that applying the
realisation $\DD_2\colon S\galg_{2}\longrightarrow\mathcal{C}^{\infty}(\R^{2})$
to the generators~\eqref{eq:gensl22} we obtain:
\begin{equation}
    \DD^{\otimes2}(\Delta(h))=q_1 p_1+q_2 p_2, 
    \quad 
    \DD^{\otimes2}(\Delta(x_{-}))=-\frac{q_1^2+q_2^2}{2},
    \quad 
    \DD^{\otimes2}(\Delta(x_{+}))=\frac{p_1^2+p_2^2}{2}.
    \label{eq:gensl22rel}
\end{equation}
That is, we obtain the following canonical two degrees of freedom realisation:
\begin{equation} 
    \begin{tikzcd}[row sep=tiny]
        (\DD_2)^{\otimes2}\colon S\galg_{2}^{\otimes 2}   \arrow{r} & 
        \mathcal{C}^{\infty}(\R^{4})
        \\ 
        \phantom{\DD_{2}\colon S\galg}(\Delta(h),\Delta(x_{-}),\Delta(x_{+})) \arrow[mapsto]{r} &
        (q_1 p_1+q_2 p_2,-(q_1^2+q_2^2)/2, (p_1^2+p_2^2)/2).
    \end{tikzcd}
    \label{eq:sl2cc2}
\end{equation}
Clearly, this is the same as an application of \Cref{lem:coprodreal}. So, in
the realisation we have the two degrees of freedom Hamiltonian:
\begin{equation}
    H_{2}^{(2)}(q_1, p_1,q_2,p_2)=H_{2}(q_1 p_1+q_2 p_2,-(q_1^2+q_2^2)/2, (p_1^2+p_2^2)/2),
\end{equation}
and the Casimir function becomes:
\begin{equation}
    C_2^{(2)}(q_1, p_1,q_2,p_2) =-L_{12}^2,
    \quad
    L_{12}=q_1 p_2-q_2 p_1.
\end{equation}
The minus sign is of course irrelevant. Note that taking into account the first
integral $C_{2}^{(2)}$ the Hamiltonian system generated by $H_2^{(2)}$ is
automatically Liouville integrable, which for $N=2$ coincides with the notion
of quasi-maximal superintegrability, being $N=2N-2$ for $N=2$. For example, for
$H_{2}=H_{\text{ho}}$ we obtain:
\begin{equation}
    H^{(2)}_{\text{ho}}(q_1, p_1,q_2,p_2) = 
    \frac{1}{2}(p_1^2+p_2^2) + 
    \frac{\omega^{2}}{2} (q_1^2+q_2^2),
\end{equation}
i.e.\ the two-dimensional isotropic harmonic oscillator. Such a system admits
additional first integrals making it maximally superintegrable. As far as we
know today, the additional integrals are not of coalgebraic origin, see for
instance~\cite{PostRiglioni2015} for a discussion of this problem in the
quantum case.

Again by direct computation, or by application of \Cref{lem:coprodreal}
we produce the following canonical $N$ degrees of freedom realisation with
Darboux coordinates $(\vb{q}, \vb{p})=(q_1,\ldots, q_N,p_{1},\ldots, p_N)$ 
\begin{equation} 
    \begin{tikzcd}[row sep=tiny]
        (\DD_{2})^{\otimes N}\colon S\galg_{2}^{\otimes N} \arrow{r} & 
        \mathcal{C}^{\infty}(\R^{2N})
        \\ 
        \phantom{\DD_{2}\colon S\galg}(\Delta^{(N)}(h),\Delta^{(N)}(x_{-}),\Delta^{(N)}(x_{+})) \arrow[mapsto]{r} &
        (\vb{q} \vdot \vb{p},-\vb{q}^2/2, \vb{p}^2/2),
    \end{tikzcd}
    \label{eq:sl2ccN}
\end{equation}
where we used the vector notation. In the same way, we obtain the following
abstract Hamiltonian:
\begin{equation}
    H_{2}^{(N)}:=\Delta(H_{2})=H_{2}(\Delta^{(N)}(h), \Delta^{(N)}(x_-), \Delta^{(N)}(x_+)) \, ,
    \label{eq:hamN}
\end{equation}
which in the realisation~\eqref{eq:sl2ccN} becomes the function:
\begin{equation}
    H_{2}^{(N)}(\vb{q}, \vb{p})=H_{2}(\vb{q} \vdot \vb{p},-\vb{q}^2/2, \vb{p}^2/2).
    \label{HN}
\end{equation}
Applying \Cref{thm:coalgfund2} we have that the above Hamiltonian
is endowed with the following left and right Casimir functions:
\begin{equation}
    C_2^{(m)}(\vb{q}, \vb{p})=-\sum_{1 \leq i<j}^m L_{ij}^2, \, 
    \quad 
    C_{2,(m)}(\vb{q}, \vb{p})=-\sum_{N-m+1 \leq i<j}^N L_{ij}^2, 
    \qquad m=2, \dots, N,
\label{LRg2}
\end{equation}
where $L_{ij}=q_i p_j-q_j p_i$ are the components of the angular momentum
tensor. Let us remark that the subscripts appearing in the left and right
Casimirs refer to the $n=2$ case of $\galg_n$.  Again, the minus sign is
irrelevant. Still following the statement of \Cref{thm:coalgfund2} the
$C_{2}^{(m)}$ and the $C_{2,(m)}$ form two sets of commuting first integrals.
We also recall that, from the general theory, $C_2^{(N)}=C_{2, (N)}$. 

Adding the Hamiltonian $H^{(N)}_{2}$ to the two sets of integrals
$L_{2}:=\set{C_2^{(2)}, \dots, C_2^{(N)}}$ and $R_{2}:=\set{C_{2,(2)}, \dots,
C_{2,(N)}}$ we obtain two sets of $N$ functionally independent commuting first
integrals:
\begin{equation}
    \widetilde{L}_{2}:=\set{H^{(N)}_{2}, C_2^{(2)}, \dots, C_2^{(N)}},
    \quad
    \widetilde{R}_{2}:=\set{H^{(N)}_{2}, C_{2, (2)}, \dots,  C_{2, (N)}},
\end{equation}
thus proving that the obtained Hamiltonian system is always Liouville
integrable, see also~\cite{Ballesteros_et_al2009}. This kind of property is
called \emph{multi-integrability}.  In total, it is easy to check, as all first
integrals except $H^{(N)}_{2}$ and $C_{2}^{(N)}=C_{2,(N)}$ are local,
that we have a set of $[(N-1)+(N-1)-1]+1=2N-2$ functionally independent
first integrals:
\begin{equation} 
  I_{2}:=\set{H^{(N)}_{2}, C_2^{(2)}, \ldots, C_2^{(N)}, C_{2,(2)},
    \ldots, C_{2,(N-1)}}.
\end{equation} 
So, the Hamiltonian system is quasi-maximally superintegrable (QMS), being just
one constant missing for maximal superintegrability. In some particular cases,
like the harmonic oscillator discussed above, there are additional functionally
independent constants leading to maximal superintegrability, see for
instance~\cite{LatiniHerranz2025} and references therein.

\subsection{The case $n=3$: $\mathfrak{g}_3$ coalgebra symmetry}

The case $n=3$ is associated to the Lie algebra
$\mathfrak{h}_6$~\cite{Zhang_et_al1990}. The coalgebra symmetry approach to
this case has been developed
in~\cite{BallesterosHerranz2001,Ballesteros_et_al2009,BlascoPhD}.

It is possible to show by direct computation that the map $\DD_{3}\colon
S\galg_{3}\longrightarrow \mathcal{C}^{\infty}(\R^{2})$ whose action on the
generating set of $S\galg_{3}$, $\set{h, x_-,x_+,y_{1,-},y_{1,+},z_{1,1}}$, is
given by:	
\begin{equation}
    h \mapsto q p \, , \quad x_-\mapsto-q^2/2, \quad x_+ \mapsto p^2/2 \, ,
    \quad y_{1,-}\mapsto-\alpha q \, , \quad y_{1,+}\mapsto \alpha p\, , \quad
    z_{1,1}\mapsto\alpha^2,
    \label{eq:realg3}
\end{equation}
where $\alpha \in \mathbb{R}$, provides a one degree of freedom canonical
symplectic realisation of $S\galg_{3}$. So, as the Hamiltonian we can take any
smooth function of the generators in the realisation~\eqref{eq:realg3}:
\begin{equation}
    \begin{aligned}
        H_{3}&=H_{3}(\DD_{3}(h),\DD_{3}(x_-),\DD_{3}(x_+),\DD_{3}(y_{1,-}),\DD_{3}(y_{1,+}),\DD_{3}(z_{1,1}))
        \\
        &=H_{3}(q p , -q^2/2, p^2/2 , -\alpha q , \alpha p, \alpha^2).
    \end{aligned}
\label{Hg3}
\end{equation}

Then, through a straightforward application of \Cref{lem:coprodreal} we obtain
a $N$ degrees of freedom canonical symplectic realisation $\DD_{3}^{\otimes
N}\colon S\galg_{3}^{\otimes N} \longrightarrow\mathcal{C}^{\infty}(\R^{2N})$
with Darboux coordinates $(\vb{q}, \vb{p})=(q_1,\ldots, q_N,p_{1},\ldots, p_N)$
whose action on the generating set of $S\galg_{3}$, $\set{h,
x_-,x_+,y_{1,-},y_{1,+},z_{1,1}}$, is given by:
\begin{equation} 
    h \mapsto \vb{q}\vdot \vb{p}, 
    \quad 
    x_-\mapsto-\vb{q}^2/2, 
    \quad 
    x_+ \mapsto \vb{p}^2/2,
    \quad 
    y_{1,-}\mapsto-\vb*{\alpha}\vdot \vb{q}, 
    \quad 
    y_{1,+}\mapsto \vb*{\alpha}\vdot \vb{p}, 
    \quad
    z_{1,1}\mapsto\vb*{\alpha}^2,
    \label{eq:realg3N}
\end{equation}
where we used the vector notation and $\vb*{\alpha}\in\R^{N}$ is a vector
of arbitrary constants. So, from the one degree of freedom Hamiltonian
we obtain the following $N$ degrees of freedom one:
\begin{equation}
    \begin{aligned}
        H^{(N)}_{3} &=\DD_{3}^{\otimes N}\bigr[H_{3}(\Delta^{(N)}(h), \Delta^{(N)}(x_-), 
      \Delta^{(N)}(x_+), \Delta^{(N)}(y_{1,-}), \Delta^{(N)}(y_{1,+}), 
       \Delta^{(N)}(z_{1,1}))\bigr]
        \\
        &=H_{3}\bigl(\vb{q}\vdot \vb{p}, -\vb{q}^2/2, \vb{p}^2/2, 
            -\vb*{\alpha}\vdot \vb{q}, \vb*{\alpha}\vdot \vb{p},
            \vb*{\alpha}^2\bigr).
	\label{HNg3}
    \end{aligned}
\end{equation}
Then, applying again \Cref{thm:coalgfund2}, we obtain the left and right
Casimir functions arising from the degree three Casimir
polynomial~\eqref{eq:C3sol}:
\begin{subequations}
    \begin{align}
	C_3^{(m)}(\vb{q}, \vb{p})&=-\sum_{1 \leq i<j<k}^m L_{ijk}^2,
        \quad
        m=3, \dots, N,
        \label{LRg3a}
        \\
        C_{3,(m)}(\vb{q}, \vb{p})&=-\sum_{N-m+1 \leq i<j<k}^N L_{ijk}^2, 
        \quad
        m=3, \dots, N,
        \label{LRg3b}
    \end{align}
    \label{LRg3}%
\end{subequations}
where 
\begin{equation}
    L_{ijk}:=\alpha_i L_{jk} - \alpha_j L_{ik}+\alpha_k L_{ij}. 
    \label{eq:buldingblock3}
\end{equation}
We observe that the index $m$ in this case starts from 3, since it easy to show
that $C_3^{(m)}$ and $C_{3,(m)}$ vanish identically for $m=1,2$. In summary, we
have that the two sets $L_{3}:=\set{C_{3}^{(m)}}_{m=3}^{N}$ and the
$R_{3}:=\set{C_{3,(m)}}_{m=3}^{N}$ are made of commuting first integrals.
Again, from the general theory we have that $C_{3}^{(N)}=C_{3,(N)}$.

So, for $N\geq3$ adding the Hamiltonian $H^{(N)}_{3}$ to the sets $L_{3}$ and
$R_{3}$ we obtain two sets composed by $N-1$ functionally independent commuting
first integrals:
\begin{equation}
    \widetilde{L}_{3}:=\set{H^{(N)}_{3}, C_3^{(3)}, \dots, C_3^{(N)}},
    \quad
    \widetilde{R}_{3}:=\set{H^{(N)}_{3}, C_{3,(3)}, \dots,  C_{3,(N)}},
\end{equation}
thus proving that the obtained Hamiltonian system is always quasi-integrable.
In total, it is easy to check, as all first
integrals except $H^{(N)}_{3}$ and $C_{3}^{(N)}=C_{3,(N)}$ are local, we have a
set of $[(N-2)+(N-2)-1]+1=2N-4$ functionally independent first integrals:
\begin{equation} 
    I_{3}:=\set{H^{(N)}_{3}, C_3^{(3)}, \ldots, C_3^{(N)}, C_{3,(3)},
    \ldots, C_{3, (N-1)}}.
\end{equation}
So, despite the Hamiltonian system is not integrable, it possesses a wealth of
first integrals. In some particular cases such that this additional first
integral exists, the Hamiltonian system immediately becomes superintegrable
with $2N-3$ first integral. These particular cases have been studied in
the literature with various methods, including the so-called subalgebra
integrability approach, see for instance~\cite{BlascoPhD,BallesterosBlasco2010}.

\subsection{The case $n=4$: $\mathfrak{g}_4$ coalgebra symmetry}

The case $n=4$ is associated to the ten-dimensional Lie algebra
$\mathfrak{g}_4$. This algebra was first introduced (in a different basis)
in~\cite{GLT_coalgebra}, where it was observed that it could be used to
interpret a $N$ degrees of freedom generalisation of the autonomous discrete
Painlev\'e equation, see~\cite{JoshiViallet2017,Grammaticosetal1991}.  However,
we observe that in~\cite{GLT_coalgebra} the details of the coalgebra
construction for such a Lie algebra were not presented.  In this Section we
fill this gap presenting an outline of such a construction.

It is possible to show by direct computation that the map $\DD_{4}\colon
S\galg_{4}\longrightarrow \mathcal{C}^{\infty}(\R^{2})$ whose action on the
generating set of $S\galg_{4}$, $\set{h,
x_-,x_+,y_{1,-},y_{1,+},z_{1,1},y_{2,-},y_{2,+},z_{1,2},z_{2,2}}$, is given by:	
\begin{equation}
    \begin{array}{ccccc}
	h \mapsto q p, & x_- \mapsto-q^2/2, & x_+ \mapsto p^2/2, & y_{1,-}\mapsto-\alpha q, & y_{1,+}\mapsto \alpha p,
        \\
	z_{1,1}\mapsto\alpha^2, & y_{2,-}\mapsto-\beta q, & y_{2,+}\mapsto\beta p, & z_{1,2}\mapsto\alpha \beta, & z_{2,2}\mapsto\beta^2,
    \end{array}
    \label{eq:realg4}
\end{equation}
where $\alpha,\beta \in \mathbb{R}$ provides a one degree of freedom canonical
symplectic realisation of $S\galg_{4}$. So, as a Hamiltonian we can take any
smooth function of the generators in the realisation~\eqref{eq:realg4}:
\begin{equation}
    \begin{aligned}
        H_{4}&=\DD_{4}\bigl[H_{4}(h,x_-,x_+,y_{1,-},y_{1,+},z_{1,1},y_{2,-},y_{2,+},z_{1,2},z_{2,2})\bigr]
        \\
        &=H_{4}\bigl(q p, -q^2/2, p^2/2, -\alpha q, \alpha p,
            \alpha^2, -\beta q, \beta p, \alpha \beta, \beta^2\bigr). 
    \end{aligned}
    \label{Hg4}
\end{equation}
Then, with a straightforward application of \Cref{lem:coprodreal} we obtain a
$N$ degrees of freedom canonical symplectic realisation $(\DD_{4})^{\otimes
N}\colon S\galg_{4}^{\otimes N}  \longrightarrow\mathcal{C}^{\infty}(\R^{2N})$
with Darboux coordinates $(\vb{q}, \vb{p})=(q_1,\ldots, q_N,p_{1},\ldots, p_N)$
whose action on the generating set of $S\galg_{4}$, $\set{h,
x_-,x_+,y_{1,-},y_{1,+},z_{1,1},y_{2,-},y_{2,+},z_{1,2},z_{2,2}}$, is given by:
\begin{equation} 
    \begin{array}{ccccc}
        h \mapsto \vb{q}\vdot \vb{p}, & 
        x_- \mapsto-\vb{q}^2/2, & 
        x_+ \mapsto \vb{p}^2/2, & 
        y_{1,-}\mapsto-\vb*{\alpha}\vdot \vb{q}, & 
        y_{1,+}\mapsto \vb*{\alpha}\vdot \vb{p},
        \\
	z_{1,1}\mapsto\vb*{\alpha}^2, & y_{2,-}\mapsto-\vb*{\beta}\vdot \vb{q}, & 
        y_{2,+}\mapsto\vb*{\beta}\vdot \vb{p}, & 
        z_{1,2}\mapsto\vb*{\alpha}\vdot \vb*{\beta}, & 
        z_{2,2}\mapsto\vb*{\beta}^2,
    \end{array}
    \label{eq:realg4N}
\end{equation}
where we used the vector notation and $\vb*{\alpha},\vb*{\beta}\in\R^{N}$ are
vectors of arbitrary constants. So, from the one degree of freedom Hamiltonian
we obtain the following $N$ degrees of freedom one:
\begin{equation}
    \begin{aligned}
        H^{(N)}_{4} &=\DD_{4}^{\otimes N} \bigl[ 
        \Delta^{(N)}(H_{4})(h,x_-,x_+,y_{1,-},y_{1,+},z_{1,1},y_{2,-},y_{2,+},z_{1,2},z_{2,2})\bigr]
        \\
        &=H_{4}\bigl(\vb{q} \cdot \vb{p},-\vb{q}^2/2, \vb{p}^2/2, -\vb*{\alpha} \cdot \vb{q}, \vb*{\alpha} \cdot \vb{p}, \vb*{\alpha}^2,-\vb*{\beta} \cdot \vb{q}, \vb*{\beta} \cdot \vb{p} ,\vb*{\alpha} \cdot \vb*{\beta},\vb*{\beta}^2 \bigr).
    \end{aligned}
    \label{HNg4}
\end{equation}
Then, applying again \Cref{thm:coalgfund2}, we obtain the left and right
Casimir functions arising from the degree four Casimir
polynomial~\eqref{eq:C4soldet}:
\begin{subequations}
    \begin{align}
	C_4^{(m)}(\vb{q}, \vb{p}) &=-\sum_{1 \leq i<j<k<l}^m L_{ijkl}^2,
        \quad m=4, \dots, N,
        \label{Lg4}
        \\
        C_{4,(m)}(\vb{q}, \vb{p}) &=-\sum_{N-m+1 \leq i<j<k<l}^N L_{ijkl}^2, 
        \quad m=4, \dots, N,
        \label{Rg4}
    \end{align}
    \label{LRg4}%
\end{subequations}
where 
\begin{equation}
     \begin{aligned}
        L_{ijkl} &= \left(\alpha _k \beta _l -\alpha_l\beta_k\right) L_{ij} 
        + \left(\alpha _l \beta _j-\alpha _j \beta _l  \right) L_{ik}
        + \left(\alpha _j \beta _k-\alpha _k \beta _j \right)L_{il}
        \\
        &+ \left(\alpha _i \beta _l-\alpha _l \beta _i \right) L_{jk} 
        + \left(\alpha_k \beta _i-\alpha _i \beta _k \right)L_{jl}
        + \left(\alpha _i \beta _j- \alpha _j\beta _i\right)L_{kl}.
     \end{aligned}
    \label{eq:buildingblocksg4}
\end{equation}
We observe again that the index $m$ in this case starts from 4, since it easy
to show that $C_4^{(m)}$ and $C_{4,(m)}$ vanish identically for $m=1,2,3$. In
summary, we have that the two sets $L_{4}:=\set{C_{4}^{(m)}}_{m=4}^{N}$ and the
$R_{4}:=\set{C_{4,(m)}}_{m=4}^{N}$ are made of commuting first integrals.
Again, from the general theory we have that $C_{4}^{(N)}=C_{4,(N)}$.

So, for $N\geq4$ adding the Hamiltonian $H^{(N)}_{4}$ to the sets $L_{4}$ and
$R_{4}$ we obtain two sets composed by $N-2$ functionally independent commuting first
integrals:
\begin{equation}
    \widetilde{L}_{4}:=\set{H_{4}^{(N)}, C_4^{(4)}, \dots, C_4^{(N)}},
    \quad
    \widetilde{R}_{4}:=\set{H_{4}^{(N)}, C_{4,(4)}, \dots,  C_{4,(N)}},
\end{equation}
thus proving that the obtained Hamiltonian system is of PLN type with rank $N-2$.
In total, it is easy to check, as all first
integrals except $H^{(N)}$ and $C_{4}^{(N)}=C_{4,(N)}$ are local, we have a
set of $[(N-3)+(N-3)-1]+1=2N-6$ functionally independent first integrals:
\begin{equation} 
    I_{4}:=\set{H_{4}^{(N)}, C_4^{(4)}, \ldots, C_4^{(N)}, C_{4,(4)},
    \ldots, C_{4,(N-1)}}.
\end{equation}
So, despite the Hamiltonian system is not integrable, it possesses several
first integrals. In some particular cases such that two additional first
integral exist, the Hamiltonian system immediately becomes superintegrable with
$2N-4$ first integrals. We defer the study of such cases to future research.

\subsection{The case $n$ generic: $\mathfrak{g}_n$ coalgebra symmetry}
\label{sss:genericn}

It is now possible to obtain the expression for generic $n$. We start with the
following result characterising a one degree of freedom symplectic realisation
of $\galg_{n}$ and the associated $N$ degrees of freedom one:

\begin{proposition}
    The Lie--Poisson algebra $S\galg_{n}$ has a one degree of freedom
    canonical symplectic realisation $\DD_{n}\colon S\galg_{n}\longrightarrow
    \mathcal{C}^{\infty}(\R^{2})$ mapping the $T_{n}$ generators of $\galg_{n}$ as:
    \begin{equation}
        \begin{array}{ccc}
            h \mapsto q p, & 
            x_- \mapsto-q^2/2,
            &
            x_+ \mapsto p^2/2,
            \\
            y_{i,-} \mapsto-\alpha^{(i)} q,
            &
            y_{i,+} \mapsto \alpha^{(i)} p, & 
            z_{i,j} \mapsto\alpha^{(i)} \alpha^{(j)},
        \end{array}
        \label{eq:realgn}
    \end{equation}
    where $\alpha^{(i)} \in \mathbb{R}$, $i=1, \dots, n-2$, are $n-2$ real
    constants. Moreover, this realisation gives rise to the following $N$
    degrees of freedom canonical symplectic realisation $\DD_{n}^{\otimes N}\colon
    S\galg_{n}^{\otimes N}\longrightarrow \mathcal{C}^{\infty}(\R^{2N})$
    mapping the generators of $\galg_{n}$ as:
    \begin{equation}
        \begin{array}{ccc}
            h \mapsto \vb{q}\vdot \vb{p}, & 
            x_- \mapsto -\vb{q}^2/2,
            &
            x_+ \mapsto \vb{p}^2/2,
            \\
            y_{i,-} \mapsto-\vb*{\alpha}^{(i)}\vdot \vb{q}, &
            y_{i,+}\mapsto \vb*{\alpha}^{(i)}\vdot \vb{p}, & 
            z_{i,j} \mapsto\vb*{\alpha}^{(i)}\vdot\vb*{\alpha}^{(j)},
        \end{array}
        \label{eq:realgnN}
    \end{equation}
    where $\vb*{\alpha}^{(i)} \in \mathbb{R}^{N}$, $i=1, \dots, n-2$, are $n-2$ real
    vectors.
    \label{prop:realgn}
\end{proposition}

\begin{remark}
    In \Cref{prop:realgn} we denoted with $\alpha^{(i)}$ the real constants
    appearing in the realization. For example, comparing with the previous case
    $n=4$,  in this notation we should identify $\alpha^{(1)}\equiv \alpha$ and
    $\alpha^{(2)}\equiv \beta$. 
    \label{rem:const}
\end{remark}

\begin{proof}
    The proof is from a direct computation. For the first part of the
    proposition, we verify that $\DD_{n}\left(\pb{y_{i,+}}{y_{j,-}}\right)
    = \pb{ \DD_{n}(y_{i,+})}{\DD_{n}(y_{j,-})}$, while the
    other relations are proven similarly:
    \begin{equation}
        \begin{aligned}
            \left\{\DD_{n}(y_{i,+}),\DD_{n}(y_{j,-})\right\}
            &=
            \left\{ \alpha^{(i)}p,-\alpha^{(j)}q \right\} 
            =-\alpha^{(i)}\alpha^{(j)}\left\{ p, q \right\} 
            \\
            &=
            \alpha^{(i)}\alpha^{(j)} =\DD_n(z_{i,j})
            =\DD_{n}\left(\left\{ y_{i,+},y_{j,-} \right\}\right).
        \end{aligned}
        \label{eq:pfyipyjm}
    \end{equation}

  \noindent   In turn, the second part of the statement follows from a direct application
    of \Cref{lem:coprodreal}. This concludes the proof of the proposition.
\end{proof}

At this point we are in place to define the corresponding $N$ degrees of
freedom Hamiltonian system. From formula~\eqref{eq:realgn} we have that as
Hamiltonian we can take any smooth function of the $T_{n}$ elements of
the generating set of $S\galg_{n}$:
\begin{equation}
    \begin{aligned}
        H_{n}&=\DD_{n}\bigl[H_{n}\bigl(h,x_-,x_+,\set{y_{i,-}}_{i},\set{y_{i,+}}_{i},\set{z_{i,j}}_{i\leq j}\bigr)\bigr]
        \\
        &=H_{n}\big(q p,-q^2/2,p^2/2,\set{-\alpha^{(i)} q}_{i},
    \set{\alpha^{(i)} p}_{i},\set{\alpha^{(i)} \alpha^{(j)}}_{i\leq j}\big), 
    \end{aligned}
    \label{eq:H1gn}
\end{equation}
where for the sake of brevity we consider understood that $i=1,\ldots,n-2$ and
$1\leq i \leq j \leq n-2$. The $N$ degrees of freedom extension of the
Hamiltonian is obtained through the application of $\DD^{\otimes N}_n$
realisation from equation~\eqref{eq:realgnN} and reads as:
\begin{equation}
    H^{(N)}=H\big(\vb{q}\vdot \vb{p},-\vb{q}^2/2,\vb{p}^2/2,
        \set{-\vb*{\alpha}^{(i)}\vdot \vb{q}}_{i},
        \set{\vb*{\alpha^{(i)}}\vdot \vb{p}}_{i},
        \set{\vb*{\alpha}^{(i)}\vdot \vb*{\alpha}^{(j)}}_{i\leq j}
        \big).
    \label{eq:HNgn}
\end{equation}

Now, let us discuss the first integral we can build for the Hamiltonian system
defined by equation~\eqref{eq:HNgn} using \Cref{thm:coalgfund2}. Along the line
of the previous results, we can generate first integrals only from the non-trivial
degree $n$ Casimir polynomial $C_{n}$. Then, upon careful direct
inspection, the left and right Casimir integrals arising from the
above-mentioned polynomial can be expressed as:
\begin{subequations}
    \begin{align}
        C^{(m)}_n(\vb{q}, \vb{p}) &=
        -\sum_{1 \leq i_1<\cdots<i_n}^m L_{i_1  \dots i_n}^2 \, , \quad m=n, \dots, N
        \label{LgN}
        \\ 
        C_{n,(m)}(\vb{q}, \vb{p}) &=
        -\sum_{N-m+1 \leq i_1<\cdots<i_n}^N L_{i_1 \dots i_n}^2,  \quad m=n, \dots, N
        \label{RgN}
    \end{align}
    \label{LRgN}
\end{subequations}
where:
\begin{equation}
    L_{i_{1},\ldots,i_{n}}
  : =
    \det
    \begin{pmatrix}
        \alpha_{i_{1}}^{(1)} & \ldots &\alpha_{i_{n}}^{(1)}
        \\
        \vdots & \ddots & \vdots
        \\
        \alpha_{i_{1}}^{(n-2)} & \ldots &\alpha_{i_{n}}^{(n-2)}
        \\
        q_{i_{1}} & \ldots & q_{i_{n}}
        \\
        p_{i_{1}} & \ldots & p_{i_{n}}
    \end{pmatrix}.
    \label{eq:Li1in}
\end{equation}
From the expression~\eqref{eq:Li1in} we see immediately that the
$C_n^{(m)}=C_{n,(m)}\equiv 0$ for $m<n$. Indeed, in such cases the building
blocks $L_{i_{1},\ldots,i_{n}}$ are identically zero because one should fill
the defining matrices with at least two equal indices, causing its determinant
to vanish.  In summary, we have that the two sets
$L_{n}:=\set{C_{n}^{(m)}}_{m=n}^{N}$ and $R_{n}:=\set{C_{n,(m)}}_{m=n}^{N}$
consist of commuting first integrals. Again, from the general theory, we have
that $C_{n}^{(N)}=C_{n,(N)}$ due to coassociativity.

\begin{remark}
    We remark that formula~\eqref{eq:Li1in} reduces to the component of the
    angular momentum $L_{i_{1},i_{2}}$ in the case $n=2$ because of the absence
    of the constants $\alpha^{(j)}_{i_{l}}$.
    \label{rem:angmom}
\end{remark}

A different yet instructive expression of the building blocks
$L_{i_{1},\ldots,i_{n}}$, can be obtained by expanding the determinant
in~\eqref{eq:Li1in} with respect to the last two rows and then resumming.
This yields the following alternative expression:
\begin{equation}
    L_{i_{1},\ldots,i_{n}}
    =
    \sum_{1\leq a < b \leq n}
    (-1)^{a+b-1} \det (A_{a,b}) L_{i_{a},i_{b}},
    \label{eq:Li1inf2}
\end{equation}
where
\begin{equation}
    A_{a,b}
    =
    \begin{pmatrix}
        \alpha_{i_{1}}^{(1)} & \ldots & \widehat{\alpha}_{i_{a}}^{(1)} 
        & \ldots &\widehat{\alpha}_{i_{b}}^{(1)} 
        & \ldots &\alpha_{i_{n}}^{(1)}
        \\
        \vdots &   & \vdots & & \vdots & & \vdots
        \\
        \alpha_{i_{1}}^{(n-2)} & \ldots & \widehat{\alpha}_{i_{a}}^{(n-2)} 
        & \ldots &\widehat{\alpha}_{i_{b}}^{(n-2)} 
        & \ldots &\alpha_{i_{n}}^{(n-2)}
    \end{pmatrix} ,
    \label{eq:Aiaib}
\end{equation}
where the hat means that the column is missing. This last form is quite
informative, since it shows that the building blocks $L_{i_1,\ldots,i_n}$ are
linear in the components of the angular momentum $L_{i_a,i_b}$ multiplied by
some degree $n-2$ polynomials in the real constants $\alpha^{(j)}_{i_{l}}$.
Moreover, formula~\eqref{eq:Li1inf2} reduces immediately to formulas
(\ref{eq:buldingblock3},\ \ref{eq:buildingblocksg4}), holding for
the cases $n=3,4$ respectively.

So, for $N\geq n$ adding the Hamiltonian $H^{(N)}_{n}$ to the sets $L_{n}$ and
$R_{n}$ we obtain two sets composed by $N-(n-2)$ functionally independent commuting first
integrals:
\begin{equation}
    \widetilde{L}_{n}:=\set{H_{n}^{(N)}, C_n^{(n)}, \dots, C_n^{(N)}},
    \quad
    \widetilde{R}_{n}:=\set{H_{n}^{(N)}, C_{n,(n)}, \dots,  C_{n,(N)}},
\end{equation}
thus proving that the obtained Hamiltonian system is of PLN type with rank
$N-(n-2)$. In total, it is easy to check, as all first integrals except
$H^{(N)}_{n}$ and $C_{n}^{(N)}=C_{n,(N)}$ are local, we have a set composed by
$[(N-n+1)+(N-n+1)-1]+1=2(N-n)+2$ functionally independent first integrals:
\begin{equation} 
    I_{n}:=\set{H_{n}^{(N)}, C_n^{(n)}, \ldots, C_n^{(N)}, C_{n,(n)},
    \ldots, C_{n, (N-1)}}.
\end{equation}
So, despite the Hamiltonian system is not integrable for $n > 2$, it admits a
large number of first integrals. In some particular cases such that $n-2$
additional first integrals exist, the Hamiltonian system immediately becomes
superintegrable with $2(N-n)+2$ first integrals. To summarise, the hierarchy of
Hamiltonians defined by the chain of Lie-Poisson algebras coming from $\galg_n$
turns out to be composed by an infinite family of integrable systems ($n= 2$),
then quasi-integrable systems ($n= 3$), and PLN type systems with rank
$N-(n-2)$ for $n>3$. 

\section{Conclusions and outlook}
\label{sec:concl}

\noindent In this paper we introduced and studied the novel $T_{n}$-dimensional
Lie algebra $\galg_n$.  This Lie algebra was suggested by two of us in a
previous work~\cite{GLT_coalgebra} when dealing with discrete integrable
systems. The main feature of this Lie algebra is being non-semisimple, but also
being a generalisation of both $\Sl_2(\K)$ and the two-photon Lie algebra
$\mathfrak{h}_6$. This algebra also has an unusually large centre
$Z(\galg_{n})$ of dimension $T_{n-2}$. 

Throughout this paper we proved the main properties of the $\galg_{n}$ Lie
algebra. In particular, in \Cref{prop:casnum} we showed that it can admit at
most one non-trivial Casimir element while in \Cref{prop:representation} we
gave a matrix representation of this algebra. Our main result, stated in
\Cref{thm:cas}, showed that such a non-trivial Casimir element exists and it is a
polynomial of degree $n$, for $n \geq 2$. Our proof was based on some elements of
representation theory, after guessing the general form of the Casimir invariant 
as a determinant of a $n \times n$ matrix through the standard method of vector fields. 
To be more specific, the key point is that we can
obtain a representation of $\Sl_{n}(\K)$ in the space of vector fields over
$\K^{n(n+1)/2}$ such that the determinant of a symmetric matrix is annihilated by its
vector fields.  This property encapsulate the well known property that the
determinant of a matrix  is a polynomial of degree $n$ and it is invariant with
respect to the action of the special linear group $\SL_{n}(\K)$. This allowed
us to conclude that the determinant of the matrix $M_{n}$ is the Casimir
polynomial we were looking for.

Finally, in the last section we used the coalgebra symmetry approach to
associate to the Lie algebra $\galg_{n}$ a hierarchy of $N$ degrees of freedom
Hamiltonian systems with Hamiltonian $H_{n}^{(N)}$~\eqref{eq:HNgn}. This was
achieved by constructing an appropriate one degree of freedom canonical
symplectic realisation and then using the known properties of the primitive
coproduct $\Delta$, see \Cref{sss:genericn}.  The construction of the
non-trivial Casimir polynomial was necessary to ensure the existence of two
sets of $N-(n-2)$ commuting first integrals, see equation~\eqref{LRgN}.  These
first integrals are expressed as sums of squares of $n$-indices objects arising
as linear combinations of the components of the angular momentum (see
\eqref{eq:Li1inf2}). We observed how the integrability properties of the
obtained Hamiltonian depend crucially on $n$, with $n=2$ being integrable,
$n=3$ quasi-integrable, and for $n\geq 4$ a Hamiltonian system of
Poincar\'e--Lyapunov--Nekhoroshev type.

The study of the Lie algebra $\galg_n$ we carried out in this paper also paves
the way to many other directions and investigations in the algebraic theory of
integrable systems. We outline some of them now, highlighting why we believe
they are important.

For instance, a natural question is the geometric interpretation of the
coalgebra symmetry with respect to the Lie algebra $\galg_n$. Indeed, for $n=2$
the coalgebra symmetry with respect to $\galg_{2}\cong\Sl_2(\R)$, in the chosen
realization, can be understood as radial symmetry, i.e.\ invariance with
respect to the Lie group of transformations $\SO(N)$. In a similar way, for
$n=3$, the coalgebra symmetry with respect to $\galg_{3}\cong\mathfrak{h}_{6}$
can be understood as symmetry around a fixed axis, i.e.\ invariance with
respect to the Lie group of transformations $\SO(N-1)\subset \SO(N)$. For
generic $n$, we don't know yet a similar interpretation. A related question is
if there is any relevant associated symmetry algebra structure for the
universal left and right invariants $C_{n}^{(m)}$ and $C_{n,(m)}$ like there is
in the case of $\galg_{2} \cong \Sl_{2}(\R)$, see~\cite{Latini2019}.
Specifically, it remains unclear whether these symmetries close in a Lie
algebra, or in a polynomial Poisson algebra.

Another area requiring further exploration is the generalisation of the
$\galg_n$ construction to other semi-direct sums with different semi-simple Lie
algebras, coming from Cartan's classification, and the study of the associated
Casimir elements. Since $\Sl_{2}(\R)\cong\Sp_{2}$ the most natural
generalisation would be to substitute the $\Sl_{2}$ Levi factor with a generic
orthosymplectic Lie algebra $\Sp_{2k}$. Let us note that the Casimir elements
of two Lie algebras having the orthosymplectic Lie algebra $\Sp_{2k}$ as Levi
factor, namely the perfect Lie algebra $w\Sp_{k}$ and the inhomogenous Lie
algebra $I\Sp_{2k}$, have been found in~\cite{Campoamor2005} as coefficients of
the characteristic polynomials of an associated matrix. Considering that a
determinant is the known term of a characteristic polynomial, it is plausible
that our results can generalise to the orthosymplectic case in similar way
as in~\cite{Campoamor2005}.

In \Cref{sec:coalg} we remarked how the coalgebra symmetry property with
respect to the Lie algebra $\galg_{n}$ in the realisation $\DD_{n}\colon
S\galg_{n}\longrightarrow \mathcal{C}^{\infty}(\R^{2})$ is not able to
automatically yield integrability as soon as $n\geq3$. To this end, we believe
that a careful application of the subalgebra integrability
method~\cite{BlascoPhD} can yield many interesting (super)integrable
Hamiltonian systems.  This method allows to build Hamiltonian with additional
first integrals using the fact that corresponding to some subalgebras one can
obtain additional non-linear Casimir invariants. Our strategy to find such
integrable case will be to classify all the optimal subalgebras of $\galg_{n}$,
i.e.\ all the different subalgebras conjugated up to adjoint action of
$G_{n}=\exp(\galg_{n})$ on $\galg_{n}$, see~\cite[Definition 3.11]{Olver1986}.
We observe that recently the slightly stricter notion of optimal
\emph{families} of subalgebras has been
introduced~\cite{Amata_etal2024,AmataOliveri2023}.  This notion is particularly
convenient for computational purposes, allowing to isolate subalgebras with
desired properties algorithmically, giving then a strong starting point to
tackle the general problem.

A final relevant open problem is a classification of the various symplectic
realisations of the Lie algebra $\galg_{n}$. Indeed, in this paper we confined
ourselves to one degree of freedom canonical symplectic realisations. However,
it is for instance known~\cite{Musso2009loopcoproducts} that the Lie algebra
$\galg_{3}\cong\mathfrak{h}_{6}$ admits a non-trivial two-degrees of freedom
canonical symplectic realisation $\DD_{3,2}\colon\galg_{3}\longrightarrow
\mathcal{C}^{\infty}(U)$, where $U=(\R^{2}_{(q_{1},q_{2})} \setminus
\set{\mu_{2}q_{1}=\mu_{1}q_{2}})\times \R^{2}_{(p_{1},p_{2})}$ and its action
on the generating set, $\set{h, x_-,x_+,y_{1,-},y_{1,+},z_{1,1}}$, is given by:
\begin{equation}
    \begin{array}{*3{>{\displaystyle}c}}
        h \mapsto q_{1} p_{1}+q_{2}p_{2}, & x_-\mapsto-\frac{q_{1}^{2}+q_{2}^2}{2}, 
        &
        x_+ \mapsto \frac{p_{1}^2+p_{2}^{2}}{2}+\frac{\kappa}{{(\mu_{2}q_{1}-\mu_{1}q_{2})}^{2}},
        \\[12pt]
    y_{1,-}\mapsto-\mu_{1} q_{1}-\mu_{2}q_{2}, & 
    y_{1,+}\mapsto \mu_{1} p_{1}+\mu_{2}p_{2}, &
    z_{1,1}\mapsto \mu_{1}^{2}+\mu_{2}^{2},
    \end{array}
    \label{eq:realg32}
\end{equation}
$\mu_1$, $\mu_2$, $\kappa$ being real parameters. It is therefore natural to
address the classification of such symplectic realisations and to investigate
their potential use in the construction of non-trivial integrable systems,
including possible applications of the loop coproduct
method~\cite{Musso2010gaudin,Musso2010loop,Musso2009loopcoproducts}. We observe
that such a method, which is an extension of the coalgebra symmetry approach
that we explained in \Cref{sss:coalg}, is a promising approach for the
construction of integrable systems, but its application was limited because of
the lack of suitable Lie algebras where the construction may apply. We believe
that the newly introduced Lie algebra $\galg_n$ can be used to fill this gap
and produce novel interesting result in such a theory.

\section*{Acknowledgements}

G.G.\ and D.L.\ acknowledge support from the INFN research project
\emph{Mathematical Methods in NonLinear Physics} (MMNLP), Gruppo~4 -- Fisica
Teorica.  Their work was also partially supported by the GNFM of the Istituto
Nazionale di Alta Matematica (INdAM). D.L.\ further acknowledges partial
funding from MUR -- Dipartimento di Eccellenza 2023--2027 (CUP~G43C22004580005,
project code DECC23\_012\_DIP).


\end{document}